\documentclass[11pt,onecolumn,draftclsnofoot,journal]{IEEEtran}

\usepackage{citesort}
\usepackage{graphicx,subfigure}
\graphicspath{{./fig/}}
\usepackage{amsmath,amssymb,amsthm,mathtools}
\usepackage{algorithm,algorithmic}
\usepackage{enumitem,multirow,array}
\usepackage{tikz,pgfplots}
\usepackage{datetime}

\usepackage{color}
\usepackage{umoline}

\newcommand{\calA}{\mathcal{A}}
\newcommand{\calS}{\mathcal{S}}
\newcommand{\calH}{\mathcal{H}}
\newcommand{\calF}{\mathcal{F}}
\newcommand{\calT}{\mathcal{T}}
\newcommand{\calR}{\mathcal{R}}
\newcommand{\calD}{\mathcal{D}}
\newcommand{\calL}{\mathcal{L}}
\newcommand{\calP}{\mathcal{P}}
\newcommand{\bbR}{\mathbb{R}}
\newcommand{\bbC}{\mathbb{C}}
\newcommand{\bbE}{\mathbb{E}}
\newcommand{\bbP}{\mathbb{P}}
\newcommand{\bfz}{\mathbf{0}}
\newcommand{\bfo}{\mathbf{1}}
\newcommand{\rmF}{\mathrm{F}}
\newcommand{\hn}{\widehat{\nabla}}
\newcommand{\hH}{\widehat{H}}
\newcommand{\norm}[1]{\|{#1}\|}
\newcommand{\T}{\top}
\newcommand{\Herm}{\mathrm{H}}

\newcommand{\diag}{\mathrm{diag}}

\newcommand{\re}{\mathrm{Re}}
\newcommand{\im}{\mathrm{Im}}

\newtheorem{theorem}{Theorem}[section]
\newtheorem{lemma}[theorem]{Lemma}
\newtheorem{corollary}[theorem]{Corollary}

\newtheorem{definition}[theorem]{Definition}

\newtheorem{claim}[theorem]{Claim}

\begin{document}

\title{Multichannel Sparse Blind Deconvolution \\on the Sphere}

\author{Yanjun~Li,~
        and~Yoram~Bresler,~\IEEEmembership{Fellow,~IEEE}%
\thanks{This work was supported in part by the National Science Foundation (NSF) under Grant IIS 14-47879. Some of the results in this paper were presented at NeurIPS 2018 \cite{li2018global}.}
\thanks{Y. Li and Y. Bresler are with the Coordinated Science Laboratory and the Department of Electrical and Computer Engineering, University of Illinois at Urbana-Champaign, Urbana, IL 61801, USA (e-mail: yli145@illinois.edu, ybresler@illinois.edu).}}

\markboth{Last revised: \currenttime,~\today}%
{Last revised: \currenttime,~\today}

\maketitle

\begin{abstract}
Multichannel blind deconvolution is the problem of recovering an unknown signal $f$ and multiple unknown channels $x_i$ from their circular convolution $y_i=x_i \circledast f$ ($i=1,2,\dots,N$). We consider the case where the $x_i$'s are sparse, and convolution with $f$ is invertible. Our nonconvex optimization formulation solves for a filter $h$ on the unit sphere that produces sparse output $y_i\circledast h$. Under some technical assumptions, we show that all local minima of the objective function correspond to the inverse filter of $f$ up to an inherent sign and shift ambiguity, and all saddle points have strictly negative curvatures. This geometric structure allows successful recovery of $f$ and $x_i$ using a simple manifold gradient descent (MGD) algorithm. Our theoretical findings are complemented by numerical experiments, which demonstrate superior performance of the proposed approach over the previous methods. 
\end{abstract}


\section{Introduction}

Blind deconvolution, which aims to recover unknown vectors $x$ and $f$ from their convolution $y = x \circledast f$, has been extensively studied, especially in the context of image deblurring \cite{kundur1996blind,cho2009fast,levin2011understanding,xu2013unnatural}. 
Recently, algorithms with theoretical guarantees have been proposed for single channel blind deconvolution \cite{ahmed2014blind,ling2015self,chi2016guaranteed,li2018rapid,lee2017blind,huang2018blind,zhang2017global}. In order for the problem to be well-posed, these previous methods assume that \emph{both} $x$ and $f$ are constrained, to either reside in a known subspace or be sparse over a known dictionary. However, these methods cannot be applied if $f$ (or $x$) is unconstrained, or does not have a subspace or sparsity structure.

In many applications in communications \cite{tong1998multichannel}, imaging \cite{she2015image}, and computer vision \cite{zhang2013multi}, convolutional measurements $y_i = x_i\circledast f$ are taken between a single signal (resp. filter) $f$ and multiple filters (resp. signals) $\{x_i\}_{i=1}^N$. We call such problems multichannel blind deconvolution (MBD).\footnote{Since convolution is a commutative operation, we use ``signal'' and ``filter'' interchangeably.} Importantly, in this multichannel setting, one can assume that only $\{x_i\}_{i=1}^N$ are structured, and $f$ is unconstrained. 
While there has been abundant work on single channel blind deconvolution (with both $f$ and $x$ constrained), research in MBD (with $f$ unconstrained) is relatively limited. Traditional MBD works assumed that the channels $x_i$'s are FIR filters \cite{tong1991new,moulines1995subspace,xu1995least} or IIR filters \cite{gurelli1995evam}, and proposed to solve MBD using subspace methods. Despite the fact that MBD with a linear (i.e., standard, non-circular) convolution model is known to have a unique solution under mild conditions \cite{harikumar1998fir}, the problem is generally ill-conditioned \cite{lee2018spectral}. Recent works improved the conditioning of such problems by introducing subspace or low-rank structures for the multiple channels \cite{lee2018spectral,lee2018fast}. 

In this paper, while retaining the unconstrained form of $f$, we consider a different structure of the multiple channels $\{x_i\}_{i=1}^N$: sparsity. The resulting problem is termed multichannel sparse blind deconvolution (MSBD). The sparsity structure arises in many real-world applications.

\textbf{Opportunistic underwater acoustics:} Underwater acoustic channels are sparse in nature \cite{berger2010sparse}. Estimating such sparse channels with an array of receivers using opportunistic sources (e.g., shipping noise) involves a blind deconvolution problem with multiple unknown sparse channels \cite{sabra2004blind,tian2017multichannel}.

\textbf{Reflection seismology:} Thanks to the layered earth structure, reflectivity in seismic signals is sparse. It is of great interest to simultaneous recover the filter (also known as the wavelet), and seismic reflectivity along the multiple propagation paths between the source and the geophones \cite{kaaresen1998multichannel}.

\textbf{Functional MRI:} Neural activity signals are composed of brief spikes and are considered sparse. However, observations via functional magnetic resonance imaging (fMRI) are distorted by convolving with the hemodynamic response function. A blind deconvolution procedure can reveal the underlying neural activity \cite{gitelman2003modeling}.

\textbf{Super-resolution fluorescence microscopy:} In super-resolution fluorescence microscopic imaging, photoswitchable probes are activated stochastically to create multiple sparse images and allow microscopy of nanoscale cellular structures \cite{rust2006sub,betzig2006imaging}. One can further improve the resolution via a computational deconvolution approach, which mitigates the effect of the point spread function (PSF) of the microscope \cite{mukamel2012statistical}. It is sometimes difficult to obtain the PSF (e.g., due to unknown aberrations), and one needs to jointly estimate the microscopic images and the PSF \cite{sarder2006deconvolution}. 

Previous approaches to MSBD have provided efficient iterative algorithms to compute maximum likelihood (ML) estimates of parametric models of the channels $\{x_i\}_{i=1}^N$ \cite{tian2017multichannel}, or maximum a posteriori (MAP) estimates in various Bayesian frameworks \cite{kaaresen1998multichannel,zhang2013multi}. However, these algorithms usually do not have theoretical guarantees or sample complexity bounds. 

Recently, guaranteed algorithms for MSBD have been developed. Wang and Chi \cite{wang2016blind} proposed a convex formulation of MSBD based on $\ell_1$ minimization, and gave guarantees for successful recovery under the condition that $f$ has one dominant entry that is significantly larger than other entries. In our previous paper \cite{li2017blind}, we solved a nonconvex formulation using projected gradient descent (truncated power iteration), and proposed an initialization algorithm to compute a sufficiently good starting point. However, in that work, theoretical guarantees were derived only for channels that are sparse with respect to a Gaussian random dictionary, but not channels that are sparse with respect to the standard basis. 

We would like to emphasize that, while earlier papers on MBD \cite{tong1991new,moulines1995subspace,xu1995least,gurelli1995evam} consider a linear convolution model, more recent guaranteed methods for MSBD \cite{wang2016blind,li2017blind} consider a circular convolution model. 
By zero padding the signal and the filter, one can rewrite a linear convolution as a circular convolution. In practice, circular convolution is often used to approximate a linear convolution when the filter has a compact support or decays fast \cite{strohmer2002four}, and the signal has finite length or satisfies a circular boundary condition \cite{cho2009fast}. The accelerated computation of circular convolution via the fast Fourier transform (FFT) is especially beneficial in 2D or 3D applications \cite{cho2009fast,sarder2006deconvolution}.  

Multichannel blind deconvolution with a circular convolution model is also related to blind gain and phase calibration \cite{li2017identifiability,balzano2007blind,bilen2014convex,ling2018self}. Suppose that a sensing system takes Fourier measurements of unknown signals and the sensors have unknown gains and phases, i.e., $\tilde{y}_i = \diag(\tilde{f})\calF x_i$, where $x_i$ are the targeted unknown sparse signals, $\calF$ is the discrete Fourier transform (DFT) matrix, and the entries of $\tilde{f}$ represent the unknown gains and phases of the sensors. The simultaneous recovery of $\tilde{f}$ and $x_i$'s is equivalent to MSBD in the frequency domain.

In this paper, we consider MSBD with circular convolution. In addition to the sparsity prior on the channels $\{x_i\}_{i=1}^N$, we impose, without loss of generality, the constraint that $f$ has unit $\ell_2$ norm, i.e., $f$ is on the unit sphere. (This eliminates the scaling ambiguity inherent in the MBD problem.) We show that our sparsity promoting objective function (maximizing the $\ell_4$ norm) has a nice geometric landscape on the the unit sphere: \textbf{(S1)} all local minima correspond to signed shifted versions of the desired solution, and \textbf{(S2)} the objective function is strongly convex in neighborhoods of the local minima, and has strictly negative curvature directions in neighborhoods of local maxima and saddle points. Similar geometric analysis has been conducted for dictionary learning \cite{sun2017complete}, phase retrieval \cite{sun2017geometric}, and single channel sparse blind deconvolution \cite{zhang2017global}. Recently, Mei et al. \cite{mei2016landscape} analyzed the geometric structure of the empirical risk of a class of machine learning problems (e.g., nonconvex binary classification, robust regression, and Gaussian mixture model). This paper is the first such analysis for MSBD. 

Properties \textbf{(S1)} and \textbf{(S2)} allow simple manifold optimization algorithms to find the ground truth in the nonconvex formulation. Unlike the second order methods in previous works \cite{sun2017complete2,sun2017geometric}, we take advantage of recent advances in the understanding of first order methods \cite{lee2017first,jin2019stochastic}, and prove that: (1) Manifold gradient descent with random initialization can escape saddle point almost surely; (2) Perturbed manifold gradient descent (PMGD) recovers a signed shifted version of the ground truth in polynomial time with high probability. This is the first guaranteed algorithm for MSBD that does \emph{not} rely on restrictive assumptions on $f$ (e.g., dominant entry \cite{wang2016blind}, spectral flatness \cite{li2017blind}), or on $\{x_i\}_{i=1}^N$ (e.g., jointly sparse, Gaussian random dictionary \cite{li2017blind}). 

Recently, many optimization methods have been shown to escape saddle points of objective functions with benign landscapes, e.g., gradient descent \cite{lee2016gradient,panageas2016gradient}, stochastic gradient descent \cite{ge2015escaping}, perturbed gradient descent \cite{jin2017escape,jin2019stochastic}, Natasha \cite{allen2017natasha,allen2018natasha}, and FastCubic \cite{agarwal2017finding}. Similarly, optimization methods over Riemannian manifolds that can escape saddle points include manifold gradient descent \cite{lee2017first}, the trust region method \cite{sun2017complete2,sun2017geometric}, and the negative curvature method \cite{goldfarb2017using}. Our main result shows that these algorithms can be applied to MSBD thanks to the favorable geometric properties of our objective function. This paper focuses on perturbed manifold gradient descent due to its simplicity, low computational complexity and memory footprint, and good empirical performance.

The advantage of the route we take in analyzing MSBD is that, by separating the global geometry of the objective function and the gradient descent algorithm, the geometric properties may be applied to analyzing other optimization algorithms. On the flip side, the readily extensible analysis yields a suboptimal bound on the number of iterations of perturbed manifold gradient descent. Very recently, Chen et al. \cite{chen2018gradient} shows that, for a phase retrieval problem of recovering $x\in \bbR^n$, vanilla gradient descent with a random initialization can achieve high accuracy with $O(\log n)$ iterations. Such low computational complexity is achieved via a careful problem-specific analysis. We believe MSBD may also benefit from a more problem-specific analysis along the lines of \cite{chen2018gradient}, and one may obtain a tighter computational complexity bound.

Our objective function -- $\ell_4$ norm maximization on the sphere -- is very similar to the kurtosis maximization approach by Shalvi and Weinstein \cite{shalvi1990new}. Their formulation assumes that the input signal follows a random distribution with nonzero kurtosis, and they show that the desired solutions are optimizers of their objective function. They proposed to optimize their blind deconvolution criteria using projected gradient descent. In contrast, our formulation is motivated by the assumption that the signals or channels are sparse. Our sample complexity bound guarantees that all local minima of our objective function correspond to desired solutions; and our perturbed manifold gradient descent algorithm is guaranteed to escape saddle points in polynomial time. Key aspects of the two works are compared in Table \ref{tab:swa}.
\begin{table}%
\caption{Comparison of this paper with Shalvi-Weinstein algorithm \cite{shalvi1990new}}
\label{tab:swa}
\centering
\begin{tabular}{|c||c|c|}
\hline
& Shalvi-Weinstein \cite{shalvi1990new} & This paper \\
\hline
\hline
Assumption & Nonzero kurtosis & Sparsity \\
\hline
Optimality & \{desired solutions\} $\subset$ \{local minima\} & \{desired solutions\} = \{local minima\} \\
\hline
Algorithm & Projected gradient descent & Perturbed manifold gradient descent \\
\hline
Guarantees & No global guarantees & \begin{tabular}{@{}c@{}}Global guarantees with \\polynomial sample complexity \\and computational complexity\end{tabular} \\
\hline
\end{tabular}
\end{table}

The rest of this paper is organized as follows. Section \ref{sec:problem} introduces the problem setting and our smooth nonconvex formulation. Section \ref{sec:geometry} shows global geometric properties of the objective function, and derives a sample complexity bound. Section \ref{sec:optimization} gives theoretical guarantees for (perturbed) manifold gradient descent. Section \ref{sec:extension} introduces several extensions of our formulation and analysis. We conduct numerical experiments to examine the empirical performance of our method in Section \ref{sec:experiment}, and conclude the paper in Section \ref{sec:conclusion} with suggestions for future work.

We use $[n]$ as a shorthand for the index set $\{1, 2, \dots, n\}$. We use $x_{(j)}$ to denote the $j$-th entry of $x\in\bbR^n$, and $H_{(jk)}$ to denote the entry of $H\in\bbR^{n\times n}$ in the $j$-th row and $k$-th column. The superscript in $h^{(t)}$ denotes iteration number in an iterative algorithm.
Throughout the paper, if an index $j \notin [n]$, then the actual index is computed as modulo of $n$. The circulant matrix whose first column is $x$ is denoted by $C_x$. We use $\delta_{jk}$ to denote the Kronecker delta ($\delta_{jk}=0$ if $j\neq k$ and $\delta_{jk}=1$ if $j = k$). The entrywise product between vectors $x$ and $y$ is denoted by $x\odot y$, and the entrywise $k$-th power of a vector $x$ is denoted by $x^{\odot k}$. We use $\norm{\cdot}$ to denote the $\ell_2$ norm (for a vector), or the spectral norm (for a matrix). We use $\re(\cdot)$ and $\im(\cdot)$ to denote the real and imaginary parts of a complex vector or matrix.


\section{MSBD on the Sphere} \label{sec:problem}

\subsection{Problem Statement}
In MSBD, the measurements $y_1, y_2, \dots, y_N \in \bbR^n$ are the circular convolutions of unknown sparse vectors $x_1, x_2, \dots, x_N \in \bbR^n$ and an unknown vector $f\in \bbR^n$, i.e., $y_i = x_i \circledast f$. In this paper, we solve for $\{x_i\}_{i=1}^n$ and $f$ from $\{y_i\}_{i=1}^N$. One can rewrite the measurement as $Y=C_f X$, where $Y=[y_1, y_2, \dots, y_N]$ and $X=[x_1, x_2, \dots, x_N]$ are $n\times N$ matrices. Without structures, one can solve the problem by choosing any invertible circulant matrix $C_f$ and compute $X=C_f^{-1}Y$. The fact that $X$ is sparse narrows down the search space. 

Even with sparsity, the problem suffers from inherent scale and shift ambiguities. Suppose $\calS_j: \bbR^n \rightarrow \bbR^n$ denotes a circular shift by $j$ positions, i.e., $\calS_j(x)_{(k)} = x_{(k-j)}$ for $j, k \in [n]$. Note that we have $y_i = x_i \circledast f = (\alpha \calS_j(x_i)) \circledast (\alpha^{-1} \calS_{-j}(f))$ for every nonzero $\alpha\in\bbR$ and $j \in [n]$. Therefore, MSBD has equivalent solutions generated by scaling and circularly shifting $\{x_i\}_{i=1}^n$ and $f$. 

Throughout this paper, we assume that the circular convolution with the signal $f$ is invertible, i.e., there exists a filter $g$ such that $f \circledast g = e_1$ (the first standard basis vector). Equivalently, $C_f$ is an invertible matrix, and the DFT of $f$ is nonzero everywhere. Since $y_i \circledast g = x_i \circledast f \circledast g = x_i$, one can find $g$ by solving the following optimization problem:
\begin{align*}
\text{(P0)} \quad &\min_{h\in\bbR^n} ~~ \frac{1}{N} \sum_{i=1}^N \norm{C_{y_i} h}_0, \quad \text{s.t.} ~~ h \neq 0.
\end{align*}
The constraint eliminates the trivial solution that is $0$. If the solution to MSBD is unique up to the aforementioned ambiguities, then the only minimizers of (P0) are $h = \alpha \calS_{j} g$ ($\alpha\neq 0$, $j\in[n]$).

\subsection{Smooth Formulation}

\begin{figure}[htbp]%
\centering
\vspace{-0.5in}
\includegraphics[width=0.23\columnwidth]{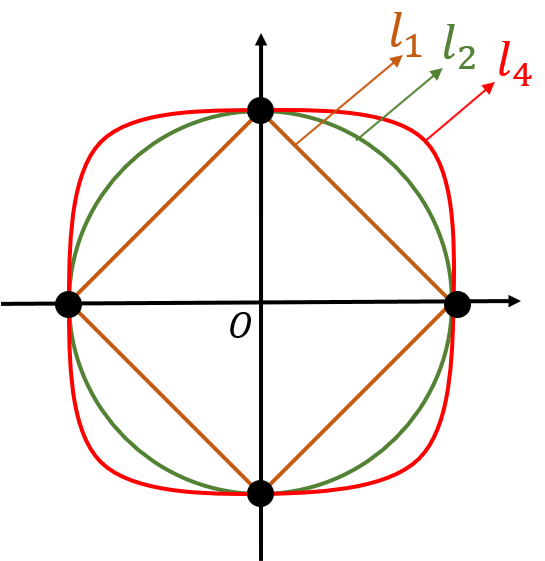}%
\caption{Unit $\ell_1$, $\ell_2$, and $\ell_4$ spheres in 2-D.}%
\label{fig:norm}
\end{figure}

Minimizing the non-smooth $\ell_0$ ``norm'' is usually challenging. Instead, one can choose a smooth surrogate function for sparsity. It is well-known that minimizing the $\ell_1$ norm can lead to sparse solutions \cite{donoho2003optimally}. An intuitive explanation is that the sparse points on the unit $\ell_2$ sphere (which we call unit sphere from now on) have the smallest $\ell_1$ norm. As demonstrated in Figure \ref{fig:norm}, these sparse points also have the largest $\ell_4$ norm. Therefore, maximizing the $\ell_4$ norm, a surrogate for the ``spikiness'' \cite{zhang2017structured} of a vector, is akin to minimizing its sparsity. In this paper, we choose the differentiable surrogate $\norm{\cdot}_4^4$ for MSBD. Recently, Bai et al. \cite{bai2018subgradient} studied dictionary learning using an $\ell_1$ norm surrogate. Applying the same technique to MSBD will be an interesting future work.

Here, we make two observations: 1) one can eliminate scaling ambiguity by restricting $h$ to the unit sphere $S^{n-1}$; 2) sparse recovery can be achieved by maximizing the ``spikiness'' $\norm{\cdot}_4^4$ \cite{zhang2017structured}. Based on these observations, we adopt the following optimization problem:
\begin{align*}
\text{(P1)} \quad &\min_{h\in\bbR^n} ~~ -\frac{1}{4N} \sum_{i=1}^N \norm{C_{y_i} R h}_4^4, \quad \text{s.t.} ~~ \norm{h} = 1.
\end{align*}
The matrix $R \coloneqq (\frac{1}{\theta n N}\sum_{i=1}^N C_{y_i}^\T C_{y_i})^{-1/2}\in\bbR^{n\times n}$ is a preconditioner, where $\theta$ is a parameter that is proportional to the sparsity level of $\{x_i\}_{i=1}^N$. In Section \ref{sec:geometry}, under specific probabilistic assumptions on $\{x_i\}_{i=1}^N$, we explain how the preconditioner $R$ works. 

Problem (P1) can be solved using first-order or second-order optimization methods over Riemannian manifolds. The main result of this paper provides a geometric view of the objective function over the sphere $S^{n-1}$ (see Figure \ref{fig:geometric}). We show that some off-the-shelf optimization methods can be used to obtain a solution $\hat{h}$ close to a scaled and circularly shifted version of the ground truth. Specifically, $\hat{h}$ satisfies $C_f R \hat{h}\approx \pm e_j$ for some $j\in[n]$, i.e., $R\hat{h}$ is approximately a signed and shifted version of the inverse of $f$. Given solution $\hat{h}$ to (P1), one can recover $f$ and $x_i$ ($i=1,2,\dots, N$) as follows:\footnote{An alternative way to recover a sparse vector $x_i$ given the recovered $\hat{f}$ and the measurement $y_i$, is to solve the non-blind deconvolution problem. For example, one can solve the sparse recovery problem $\min_x \frac{1}{2}\norm{C_{\hat{f}} x - y_i}^2 + \lambda \norm{x}_1$ using FISTA \cite{beck2009fast}. We omit the analysis of such a solution in this paper, and focus on the simple reconstruction $\hat{x}_i = C_{y_i} R \hat{h}$.}
\begin{align}
& \hat{f} = \calF^{-1}\bigl[\calF(R\hat{h})^{\odot -1}\bigr],  \label{eq:fhat}\\
& \hat{x}_i = C_{y_i} R \hat{h}.  \label{eq:xhat}
\end{align}


\section{Global Geometric View} \label{sec:geometry}

\subsection{Main Result}
In this paper, we assume that $\{x_i\}_{i=1}^N$ are random sparse vectors, and $f$ is invertible:
\begin{itemize}[leftmargin=10mm]
	\itemsep0em
	\item[(A1)] The channels $\{x_i\}_{i=1}^N$ follow a Bernoulli-Rademacher model. More precisely, $x_{i(j)}=A_{ij}B_{ij}$, where $\{A_{ij}, B_{ij}\}_{i\in[N], j\in[n]}$ are independent random variables, $B_{ij}$'s follow a Bernoulli distribution $\mathrm{Ber}(\theta)$, and $A_{ij}$'s follow a Rademacher distribution (taking values $1$ and $-1$, each with probability $1/2$). 
	\item[(A2)] The circular convolution with the signal $f$ is invertible. We use $\kappa$ to denote the condition number of $f$, which is defined as $\kappa \coloneqq \frac{\max_j |(\calF f)_{(j)}|}{\min_k |(\calF f)_{(k)}|}$, i.e., the ratio of the largest and smallest magnitudes of the DFT. This is also the condition number of the circulant matrix $C_f$, i.e. $\kappa = \frac{\sigma_1(C_f)}{\sigma_n(C_f)}$.
\end{itemize}
The Bernoulli-Rademacher model is a special case of the Bernoulli-subgaussian models. We conjecture that the derivation in this paper can be repeated for other subgaussian nonzero entries, with different tail bounds. We use the Rademacher distribution for simplicity.


Let $\phi(x) = -\frac{1}{4}\norm{x}_4^4$. Its gradient and Hessian are defined by $\nabla_\phi(x)_{(j)} = -x_j^3$, and $H_\phi(x)_{(jk)} = -3 x_j^2 \delta_{jk}$. Then the objective function in (P1) is
\[
L(h) = \frac{1}{N}\sum_{i=1}^N \phi(C_{y_i}Rh),
\]
where $R = (\frac{1}{\theta n N}\sum_{i=1}^N C_{y_i}^\T C_{y_i})^{-1/2}$. The gradient and Hessian are
\begin{align*}
& \nabla_L(h) = \frac{1}{N}\sum_{i=1}^N R^\T C_{y_i}^\T \nabla_\phi(C_{y_i}Rh), \\
& H_L(h) = \frac{1}{N}\sum_{i=1}^N R^\T C_{y_i}^\T H_\phi(C_{y_i}Rh) C_{y_i}R.
\end{align*}
Since $L(h)$ is to be minimized over $S^{n-1}$, we use optimization methods over Riemannian manifolds \cite{absil2009optimization}. To this end, we define the tangent space at $h\in S^{n-1}$ as $\{z\in\bbR^n : z\perp h\}$ (see Figure \ref{fig:tangent}). 
We study the Riemannian gradient and Riemannian Hessian of $L(h)$ (gradient and Hessian along the tangent space at $h\in S^{n-1}$):
\begin{align}
\begin{aligned}
& \hn_L(h) = P_{h^\perp} \nabla_L(h), \\
& \hH_L(h) = P_{h^\perp} H_L(h) P_{h^\perp} - \langle \nabla_L(h), h \rangle P_{h^\perp},
\end{aligned}
\label{eq:Riemannian_gradient_Hessian}
\end{align}
where $P_{h^\perp} = I - hh^\T$ is the projection onto the tangent space at $h$. We refer the readers to \cite{absil2009optimization} for a more comprehensive discussion of these concepts.
\begin{figure}[htbp]%
\centering
\includegraphics[width=0.2\columnwidth]{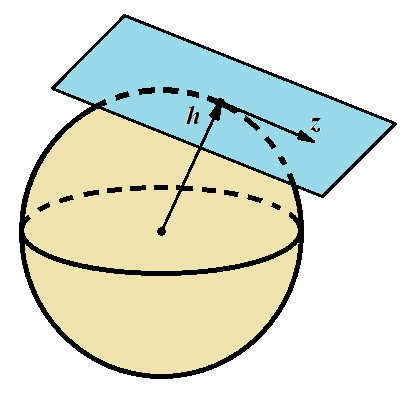}%
\caption{A demonstration of the tangent space of $S^{n-1}$ at $h$, the origin of which is translated to $h$. The Riemannian gradient and Riemannian Hessian are defined on tangent spaces.}%
\label{fig:tangent}%
\end{figure}

\begin{figure}[htbp]%
\centering
\subfigure[]{
	\includegraphics[width=0.3\columnwidth]{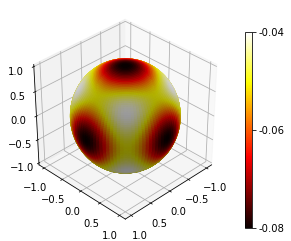}%
	\label{fig:V}
}
\subfigure[]{
	\includegraphics[width=0.3\columnwidth]{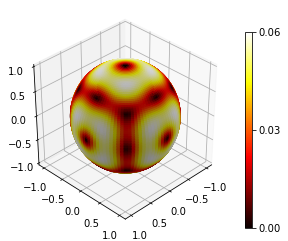}%
	\label{fig:G}
}
\subfigure[]{
	\includegraphics[width=0.3\columnwidth]{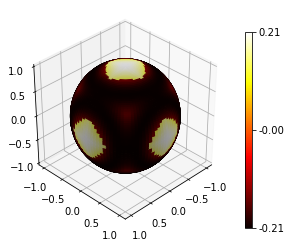}%
	\label{fig:H}
}
\caption{Geometric structure of the objective function over the sphere. For $n=3$, we plot the following quantities on the sphere $S^{2}$: (a) $\bbE L''(h)$, (b) $\norm{\bbE \hn_{L''}(h)}$, and (c) $\min_{z\perp h, \norm{z}=1} z^\T\bbE \hH_{L''}(h) z$.}%
\label{fig:geometric}%
\end{figure}

The toy example in Figure \ref{fig:geometric} demonstrates the geometric structure of the objective function on $S^{n-1}$. (As shown later, the quantity $\bbE L''(h)$ is, up to an unimportant rotation of the coordinate system, a good approximation to $L(h)$.) The local minima correspond to signed shifted versions of the ground truth (Figure \ref{fig:V}). The Riemannian gradient is zero at stationary points, including local minima, saddle points, and local maxima of the objective function when restricted to the sphere $S^{n-1}$. (Figure \ref{fig:G}). The Riemannian Hessian is positive definite in the neighborhoods of local minima, and has at least one strictly negative eigenvalue in the neighborhoods of local maxima and saddle points (Figure \ref{fig:H}). We say that a stationary point is a ``strict saddle point'' if the Riemannian Hessian has at least one strictly negative eigenvalue. The Riemannian Hessian is  negative definite in the neighborhood of a local maximum. Hence, local maxima are strict saddle points. Our main result Theorem \ref{thm:main} formalizes the observation that $L(h)$ only has two types of stationary points: 1) local minima, which are close to signed shifted versions of the ground truth, and 2) strict saddle points.

\begin{theorem} \label{thm:main}
Suppose Assumptions (A1) and (A2) are satisfied, and the Bernoulli probability satisfies $\frac{1}{n} \leq \theta < \frac{1}{3}$. Let $\kappa$ be the condition number of $f$, let $\rho < 10^{-3}$ be a small tolerance constant. There exist constants $c_1, c_1' > 0$ (depending only on $\theta$), such that: if $N > \max\{\frac{c_1 n^9}{\rho^4} \log \frac{n}{\rho}, \frac{c_1 \kappa^8 n^8}{\rho^4} \log n\}$, then with probability at least $1 - n ^{-c_1'}$, every local minimum $h^*$ in (P1) is close to a signed shifted version of the ground truth. I.e., for some $j\in[n]$:
\[
\norm{C_f R h^* \pm e_{j}} \leq 2\sqrt{\rho}.
\]
Moreover, one can partition $S^{n-1}$ into three sets $\calH_1$, $\calH_2$, and $\calH_3$ that satisfy (for some $c(n,\theta,\rho)>0$):
\begin{enumerate}
	\itemsep0em
	\item[$\circ$] $L(h)$ is strongly convex in $\calH_1$, i.e., the Riemannian Hessian is positive definite:
	\[
	\min_{\substack{z: \norm{z}=1\\z\perp h}} z^\T \hH_{L}(h) z \geq c(n,\theta,\rho) > 0.
	\]

	\item[$\circ$] $L(h)$ has negative curvature in $\calH_2$, i.e., the Riemannian Hessian has a strictly negative eigenvalue:
	\[
	\min_{\substack{z: \norm{z}=1\\z\perp h}} z^\T \hH_{L}(h) z \leq -c(n,\theta,\rho) < 0.
	\]
	
	\item[$\circ$] $L(h)$ has a descent direction in $\calH_3$, i.e., the Riemannian gradient is nonzero: 
	\[
	\norm{\hn_{L}(h)} \geq c(n,\theta,\rho) > 0.
	\]
\end{enumerate}
Clearly, all the stationary points of $L(h)$ on $S^{n-1}$ belong to $\calH_1$ or $\calH_2$. The stationary points in $\calH_1$ are local minima, and the stationary points in $\calH_2$ are strict saddle points.
The sets $\calH_1$, $\calH_2$, $\calH_3$ are defined in \eqref{eq:rotation_set}, and the positive number $c(n,\theta, \rho)$ is defined in \eqref{eq:gap_constant}.
\end{theorem}

The sample complexity bound on the number $N$ of channels in Theorem \ref{thm:main} is pessimistic. The numerical experiments in Section \ref{sec:synthetic} show that successful recovery can be achieved with a much smaller $N$ (roughly linear in $n$). One reason for this disparity is that Theorem \ref{thm:main} is a worst-case theoretical guarantee that holds uniformly for all $f$. For a specific $f$, the sampling requirement can be lower. Another reason is that Theorem \ref{thm:main} guarantees favorable geometric properties of the objective function, which is sufficient but not necessary for successful recovery. Moreover, our proof uses several uniform upper bounds on $\{x_i\}_{i=1}^N$, which are loose in an average sense. Interested readers may try to tighten those bounds and obtain a less demanding sample complexity.

We only consider the noiseless case in Theorem \ref{thm:main}. One can extend our analysis to noisy measurements by bounding the perturbation of the objective function caused by noise. By Lipschitz continuity, small noise will incur a small perturbation on the Riemannian gradient and the Riemannian Hessian. As a result, favorable geometric properties will still hold under low noise levels. We omit detailed derivations here. In Section \ref{sec:experiment}, we verify by numerical experiments that the formulation in this paper is robust against noise.


\subsection{Proof of the Main Result} 

Since $\bbE \frac{1}{\theta n}  C_{y_i}^\T C_{y_i} = C_f^\T C_f$, by the law of large numbers, $R = (\frac{1}{\theta n N}\sum_{i=1}^N C_{y_i}^\T C_{y_i})^{-1/2}$ asymptotically converges to $(C_f^\T C_f)^{-1/2}$ as $N$ increases (see the proof of Lemma \ref{lem:not_flat}). Therefore, $L(h)$ can be approximated by
\begin{align*}
L'(h) = \frac{1}{N}\sum_{i=1}^N \phi(C_{y_i}(C_f^\T C_f)^{-1/2}h)=\frac{1}{N}\sum_{i=1}^N \phi(C_{x_i}C_f(C_f^\T C_f)^{-1/2}h).
\end{align*}
Since $C_f(C_f^\T C_f)^{-1/2}$ is an orthogonal matrix, one can study the following objective function by rotating on the sphere $h' = C_f(C_f^\T C_f)^{-1/2} h$:
\begin{align*}
L''(h') = \frac{1}{N}\sum_{i=1}^N \phi(C_{x_i} h').
\end{align*}
Our analysis consists of three parts: 1) geometric structure of $\bbE L''$, 2) deviation of $L''$ (or its rotated version $L'$) from its expectation $\bbE L''$, and 3) difference between $L$ and $L'$.


\textbf{Geometric structure of $\bbE L''$.} By the Bernoulli-Rademacher model (A1), the Riemannian gradient for $h\in S^{n-1}$ is computed as
\begin{align}
& \bbE \hn_{L''}(h) = P_{h^\perp} \bbE \nabla_{L''}(h) = n\theta(1-3\theta)(\norm{h}_4^4 \cdot h - h^{\odot 3}).  \label{eq:gradient}
\end{align}
The Riemannian Hessian is
\begin{align}
& \bbE \hH_{L''}(h) = P_{h^\perp} \bbE H_{L''}(h) P_{h^\perp} - h^\T \bbE \nabla_{L''}(h) \cdot P_{h^\perp}   \nonumber\\
& \qquad = n\theta(1-3\theta) \bigl[\norm{h}_4^4 \cdot I + 2\norm{h}_4^4 \cdot hh^\T - 3\cdot \diag(h^{\odot 2}) \bigr].  \label{eq:Hessian}
\end{align}
Details of the derivation of \eqref{eq:gradient} and \eqref{eq:Hessian} can be found in Appendix \ref{app:g_n_H}.

At a stationary point of $\bbE L''(h)$ on $S^{n-1}$, the Riemannian gradient is zero. Since
\begin{align}
& \norm{\bbE \hn_{L''}(h)} = n\theta(1-3\theta) \sqrt{\norm{h}_6^6 - \norm{h}_4^8}  \nonumber\\
& = n\theta(1-3\theta) \sqrt{\sum_{1\leq j < k \leq n} h_{(j)}^2h_{(k)}^2(h_{(j)}^2 - h_{(k)}^2)^2}, \label{eq:gradient_norm}
\end{align}
all nonzero entries of a stationary point $h_0$ have the same absolute value. Equivalently, $h_{0(j)}=\pm 1/\sqrt{r}$ if $j\in \Omega$ and $h_{0(j)}=0$ if $j\notin \Omega$, for some $r\in[n]$ and $\Omega \subset [n]$ such that $|\Omega|=r$. Without loss of generality (as justified below), we focus on stationary points that satisfy $h_{0(j)}=1/\sqrt{r}$ if $j \in \{1,2,\dots, r\}$ and $h_{0(j)}=0$ if $j \in \{r+1,\dots, n\}$. The Riemannian Hessian at these stationary points is
\begin{align}
\bbE \hH_{L''}(h_0) = \frac{n\theta(1-3\theta)}{r}
\begin{bmatrix} 
\frac{2}{r}\bfo_{r\times r} - 2I_r & \bfz_{r\times (n-r)}\\
\bfz_{(n-r)\times r} & I_{n-r}
\end{bmatrix}.   \label{eq:Hessian_stationary}
\end{align}

When $r = 1$, $h_0=[1, 0, 0, \dots, 0]^\T$, we have $\bbE \hH_{L''}(h_0) = n\theta(1-3\theta) P_{h_0^\perp}$. This Riemannian Hessian is positive definite on the tangent space, 
\begin{align}
\min_{\substack{z: \norm{z}=1\\z\perp h_0}} z^\T \bbE \hH_{L''}(h_0) z = n\theta(1-3\theta) > 0. \label{eq:PD}
\end{align}
Therefore, stationary points with one nonzero entry are local minima.

When $r > 1$, the Riemannian Hessian has at least one strictly negative eigenvalue:
\begin{align}
\min_{\substack{z: \norm{z}=1\\z\perp h_0}} z^\T \bbE \hH_{L''}(h_0) z = -\frac{2n\theta(1-3\theta)}{r} < 0. \label{eq:NC}
\end{align}
Therefore, stationary points with more than one nonzero entry are strict saddle points, which, by definition, have at least one negative curvature direction on $S^{n-1}$.
One such negative curvature direction satisfies $z_{(1)} = (r-1)/\sqrt{r(r-1)}$, $z_{(j)} = -1/\sqrt{r(r-1)}$ for $j\in\{2,3,\dots, r\}$, and $z_{(j)} = 0$ for $j\in\{r+1,\dots, n\}$.

The Riemannian Hessian at other stationary points (different from the above stationary points by permutations and sign changes) can be computed similarly. By \eqref{eq:Hessian}, a permutation and sign changes of the entries in $h_0$ has no effect on the bounds in \eqref{eq:PD} and \eqref{eq:NC}, because the eigenvector $z$ that attains the minimum undergoes the same permutation and sign changes as $h_0$.

Next, in Lemma \ref{lem:partition}, we show that the properties of positive definiteness and negative curvature not only hold at the stationary points, but also hold in their neighborhoods defined as follows.
\begin{definition} \label{def:close}
We say that a point $h$ is in the $(\rho,r)$-neighborhood of a stationary point $h_0$ of $\bbE L''(h)$ with $r$ nonzero entries, if $\norm{h^{\odot 2} - h_0^{\odot 2}}_\infty \leq \frac{\rho}{r}$. We define three sets: 
\begin{align*}
& \calH''_1 \coloneqq \{\text{Points in the $(\rho,1)$-neighborhoods of stationary points with $1$ nonzero entry}\}, \\
& \calH''_2 \coloneqq \{\text{Points in the $(\rho,r)$-neighborhoods of stationary points with $r>1$ nonzero entries}\}, \\
& \calH''_3 \coloneqq S^{n-1} \backslash (\calH''_1\cup \calH''_2).
\end{align*}
Clearly, $\calH''_1 \cap \calH''_2 =\emptyset$ for $\rho < 1/3$, hence $\calH''_1$, $\calH''_2$, and $\calH''_3$ form a partition of $S^{n-1}$.
\end{definition}

\begin{lemma} \label{lem:partition}
Assume that positive constants $\theta < 1/3$, and $\rho < 10^{-3}$. Then
\begin{enumerate}
	\itemsep0em
	\item[$\circ$] For $h \in \calH''_1$,
\begin{align}
\min_{\substack{z: \norm{z}=1\\z\perp h}} z^\T \bbE \hH_{L''}(h) z \geq n\theta(1-3\theta)(1-24\sqrt{\rho}) > 0. \label{eq:PD2}
\end{align}
	
	\item[$\circ$] For $h \in \calH''_2$,
\begin{align}
\min_{\substack{z: \norm{z}=1\\z\perp h}} z^\T \bbE \hH_{L''}(h) z \leq -\frac{n\theta(1-3\theta)(2-24\sqrt{\rho})}{r} < 0. \label{eq:NC2}
\end{align}
	
	\item[$\circ$] For $h \in \calH''_3$,
\begin{align}
\norm{\bbE \hn_{L''}(h)} \geq \frac{\theta(1-3\theta)\rho^2}{n} > 0. \label{eq:NG2}
\end{align}
\end{enumerate}
\end{lemma}

Lemma \ref{lem:partition}, and all other lemmas, are proved in the Appendix.


\textbf{Deviation of $L''$ from $\bbE L''$}. As the number $N$ of channels increases, the objective function $L''$ asymptotically converges to its expected value $\bbE L''$. Therefore, we can establish the geometric structure of $L''$ based on its similarity to $\bbE L''$. To this end, we give the following result.

\begin{lemma} \label{lem:concentration}
Suppose that $\theta < 1/3$. There exist constants $c_2, c_2' > 0$ (depending only on $\theta$), such that: if $N > \frac{c_2 n^9}{\rho^4}\log \frac{n}{\rho}$, then with probability at least $1 - e ^{-c_2'n}$,
\begin{align*}
\sup_{h\in S^{n-1}} \norm{\hn_{L''}(h) - \bbE\hn_{L''}(h)} \leq \frac{\theta(1-3\theta)\rho^2}{4n},
\end{align*}
\begin{align*}
\sup_{h\in S^{n-1}} \norm{\hH_{L''}(h) - \bbE \hH_{L''}(h)} \leq \frac{\theta(1-3\theta)\rho^2}{n}.
\end{align*}
\end{lemma}

By Lemma \ref{lem:concentration}, the deviations from the corresponding expected values of the Riemannian gradient and Hessian due to a finite number of random $x_i$'s are small compared to the bounds in Lemma \ref{lem:partition}. Therefore, the Rimannian Hessian of $L''$ is still positive definite in the neighborhood of local minima, and has at least one strictly negative eigenvalue in the neighborhood of strict saddle points; and the Riemannian gradient of $L''$ is nonzero for all other points on the sphere. Since $L'$ and $L''$ differ only by an orthogonal matrix transformation of their argument, the geometric structure of $L'$ is identical to that of $L''$ up to a rotation on the sphere.


\textbf{Difference between $L$ and $L'$}. Recall that $L$ asymptotically converges to $L'$ as $N$ increases. The following result bounds the difference for a finite $N$. 

\begin{lemma} \label{lem:not_flat}
Suppose that $\frac{1}{n}\leq \theta < \frac{1}{3}$. There exist constants $c_3, c_3' > 0$ (depending only on $\theta$), such that: if $N > \frac{c_3 \kappa^8 n^8}{\rho^4}\log n$, then with probability at least $1 - n ^{-c_3'}$,
\begin{align*}
\sup_{h\in S^{n-1}} \norm{\hn_{L}(h) - \hn_{L'}(h)} \leq \frac{\theta(1-3\theta)\rho^2}{4n},
\end{align*}
\begin{align*}
\sup_{h\in S^{n-1}} \norm{\hH_{L}(h) - \hH_{L'}(h)} \leq \frac{\theta(1-3\theta)\rho^2}{n}.
\end{align*}
\end{lemma}


We use $(C_f^\T C_f)^{1/2} C_f^{-1}\calH = \{(C_f^\T C_f)^{1/2} C_f^{-1}h: h\in\calH\}$ to denote the rotation of a set $\calH$ by the orthogonal matrix $(C_f^\T C_f)^{1/2} C_f^{-1}$. Define the rotations of $\calH''_1$, $\calH''_2$, and $\calH''_3$:
\begin{align}
\calH_1 \coloneqq (C_f^\T C_f)^{1/2} C_f^{-1}\calH''_1,\quad
\calH_2 \coloneqq (C_f^\T C_f)^{1/2} C_f^{-1}\calH''_2,\quad
\calH_3 \coloneqq (C_f^\T C_f)^{1/2} C_f^{-1}\calH''_3.
\label{eq:rotation_set}
\end{align}

Combining Lemmas \ref{lem:partition}, \ref{lem:concentration}, and \ref{lem:not_flat}, and the rotation relation between $L'$ and $L''$, we have:

\begin{enumerate}
	\itemsep0em
	\item[$\circ$] For $h \in \calH_1$, the Riemannian Hessian is positive definite:
\begin{align*}
\min_{\substack{z: \norm{z}=1\\z\perp h}} z^\T \hH_{L}(h) z \geq n\theta(1-3\theta)(1-24\sqrt{\rho} - \frac{2\rho^2}{n^2}) > 0.
\end{align*}
	
	\item[$\circ$] For $h \in \calH_2$,  the Riemannian Hessian has a strictly negative eigenvalue:
\begin{align}
\min_{\substack{z: \norm{z}=1\\z\perp h}} z^\T \hH_{L}(h) z \leq -\frac{n\theta(1-3\theta)(2-24\sqrt{\rho} - 2r\rho^2/n^2)}{r} < 0.
\label{eq:negative_curvature}
\end{align}
	
	\item[$\circ$] For $h \in \calH_3$, the Riemannian gradient is nonzero:
\begin{align*}
\norm{\hn_{L}(h)} \geq \frac{\theta(1-3\theta)\rho^2}{2n} > 0. 
\end{align*}
\end{enumerate}
Clearly, all the local minima of $L(h)$ on $S^{n-1}$ belong to $\calH_1$, and all the other stationary points are strict saddle points and belong to $\calH_2$. The bounds in Theorem \ref{thm:main} on the Riemannian Hessian and the Riemannian gradient follows by setting
\begin{align}
c(n,\theta,\rho) \coloneqq \frac{\theta(1-3\theta)\rho^2}{2n}. 
\label{eq:gap_constant}
\end{align}

We complete the proof of Theorem \ref{thm:main} by giving the following result about $\calH_1$.
\begin{lemma}\label{lem:neighborhood_local_minima}
If $h^* \in \calH_1$, then for some $j\in [n]$,
\[
\norm{C_f R h^* \pm e_j} \leq 2 \sqrt{\rho}.
\]
\end{lemma}


\section{Optimization Method} \label{sec:optimization}

\subsection{Guaranteed First-Order Optimization Algorithm}

Second-order methods over a Riemannian manifold are known to be able to escape saddle points, for example, the trust region method \cite{boumal2018global}, and the negative curvature method \cite{goldfarb2017using}. Recent works proposed to solve dictionary learning \cite{sun2017complete2}, and phase retrieval \cite{sun2017geometric} using these methods, without any special initialization schemes. Thanks to the geometric structure (Section \ref{sec:geometry}) and the Lipschitz continuity of the objective function for our multichannel blind deconvolution formulation (Section \ref{sec:problem}), these second-order methods can recover signed shifted versions of the ground truth without special initialization.

Recently, first-order methods have been shown to escape strict saddle points with random initialization \cite{lee2016gradient,panageas2016gradient}. In this paper, we use the manifold gradient descent (MGD) algorithm studied by Lee et al. \cite{lee2017first}. One can initialize the algorithm with a random $h^{(0)}$, and use the following iterative update:
\begin{align}
h^{(t+1)} = \calA(h^{(t)}) \coloneqq P_{S^{n-1}}\bigl( h^{(t)} - \gamma \hn_L(h^{(t)}) \bigr), \label{eq:mgd}
\end{align}
where $\hn_L(h)$ is the Riemannian gradient in \eqref{eq:Riemannian_gradient_Hessian}.
Each iteration takes a Riemannian gradient descent step in the tangent space, and does a retraction by normalizing the iterate (projecting onto $S^{n-1}$).

Using the geometric structure introduced in Section \ref{sec:geometry}, and some technical results in \cite{lee2017first}, the following theorem shows that, for our formulation of MSBD, MGD with a random initialization can escape saddle points almost surely.

\begin{theorem} \label{thm:mgd}
Suppose that the geometric structure in Theorem \ref{thm:main} is satisfied. If manifold gradient descent \eqref{eq:mgd} is initialized with a random $h^{(0)}$ drawn from a uniform distribution on $S^{n-1}$, and the step size is chosen as $\gamma<\frac{1}{64n^3}$, then \eqref{eq:mgd} does not converge to any saddle point of $L(h)$ on $S^{n-1}$ almost surely. 
\end{theorem}

Despite the fact that MGD almost always escapes saddle points, and that the objective function monotonically decreases (which we will prove in Lemma \ref{lem:gradient_descent_step}), it lacks an upper bound on the number of iterations required to converge to a good solution.

Inspired by recent work \cite{jin2019stochastic}, we make a small adjustment to MGD: when the Riemannian gradient becomes too small, we add a perturbation term to the iterate (see Algorithm \ref{alg:pmgd} and Figure \ref{fig:perturbation}). This allows an iterate to escape ``bad regions'' near saddle points.

\begin{algorithm}[ht]
\caption{Perturbed Manifold Gradient Descent (PMGD)}
\begin{algorithmic} \label{alg:pmgd}
\STATE {\bfseries Input:} $h^{(0)}\in\bbR^{n}$
\STATE {\bfseries Output:} $\{h^{(t)}\}_{t=1}^T$
\STATE {\bfseries Parameters:} step size $\gamma$, perturbation radius $\calD$, time interval $\calT$, and tolerance $c(n,\theta,\rho)$
\STATE $t_\mathrm{perturb} \leftarrow 0$
\FOR{$t=0,1,\dots, T-1$}
\IF{$\norm{\hn_L(h^{(t)})} < c(n,\theta,\rho)$ and $t - t_\mathrm{perturb}>\calT$}
\STATE $h^{(t)} \leftarrow \sqrt{1-\norm{z_p}^2} h^{(t)} + z_p$ where  $z_p \sim \mathrm{Uniform}\bigl(\{z: z\perp h^{(t)}, z \in \calD \calS^{n-1}\}\bigr)$
\STATE $t_\mathrm{perturb} \leftarrow t$
\ENDIF
\STATE $h^{(t+1)} \leftarrow P_{S^{n-1}}\bigl( h^{(t)} - \gamma \hn_L(h^{(t)}) \bigr)$
\ENDFOR
\end{algorithmic}
\end{algorithm}

\begin{figure}[htbp]
\centering
\includegraphics[width=0.5\columnwidth]{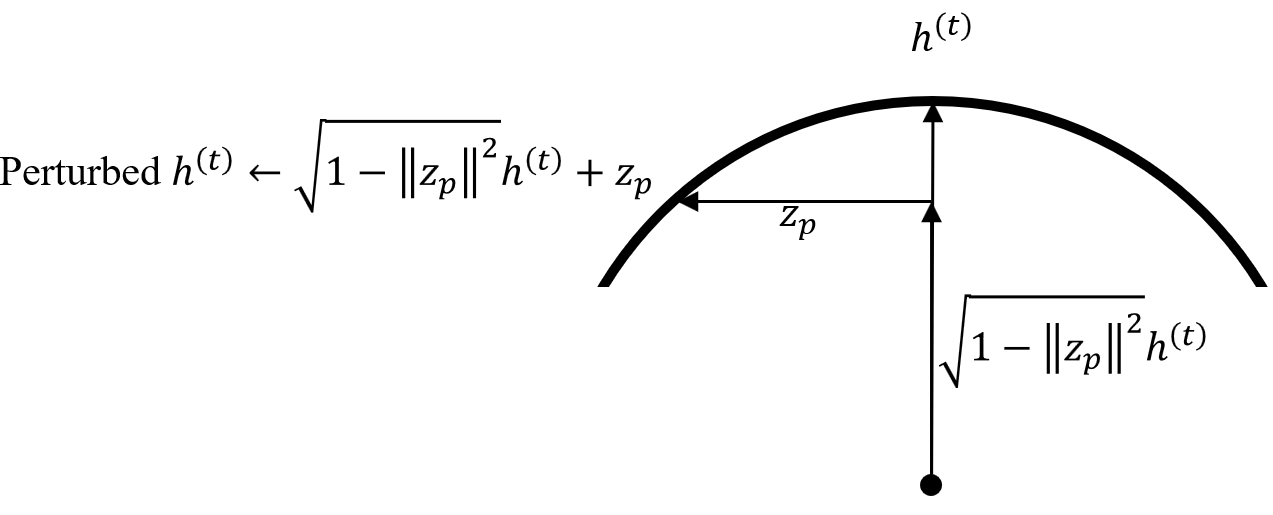}%
\caption{An illustration of the perturbation step in PMGD (Algorithm \ref{alg:pmgd}).}%
\label{fig:perturbation}%
\end{figure}

We have the following theoretical guarantee for Algorithm \ref{alg:pmgd}.

\begin{theorem}\label{thm:pmgd}
For a sufficiently large absolute constant $\xi$, if one chooses 
\[
\gamma \coloneqq \frac{1}{128n^3}, \quad \calD \coloneqq \frac{\theta^2(1-3\theta)^2}{\xi^2 n^6\log^2 n}, \quad
\calT \coloneqq \frac{\xi n^3\log n}{25\theta(1-3\theta)}, \quad c(n,\theta,\rho) \coloneqq \frac{\theta(1-3\theta)\rho^2}{2n},
\]
then with high probability, we have at least one iterate that satisfies $h^{(t)}\in\calH_1$ after the following number of iterations:
\[
T = \frac{5000n^8}{\theta^2(1-3\theta)^2\rho^4} + \frac{\xi^4n^{12}\log^4n}{800\theta^4(1-3\theta)^4}.
\]
\end{theorem}

Theorem \ref{thm:pmgd} shows that, PMGD outputs, in polynomial time, a solution $h^{(t)}\in\calH_1$ that is close to a signed and shifted version of the ground truth. 
One can further bound the recovery error of the signal and the channels as follows.

\begin{corollary}\label{cor:f_and_x}
Under the same setting as in Theorem \ref{thm:pmgd}, the recovered $\hat{f}$ and $\hat{x}_i$ ($i = 1,2,\dots, N$) in \eqref{eq:fhat} and \eqref{eq:xhat}, computed using the output $\hat{h} = h^{(t)} \in \calH_1$ of PMGD, satisfy
\begin{align*}
& \frac{\norm{\hat{x}_i \pm \calS_j(x_i)}}{\norm{x_i}} \leq 2 \sqrt{\rho n},\\
& \frac{\norm{\hat{f} \pm S_{-j}(f)}}{\norm{f}} \leq \frac{2 \sqrt{\rho n}}{1 - 2 \sqrt{\rho n}}.
\end{align*}
for some $j\in [n]$.
\end{corollary}

Similar to our requirement on $N$ in Theorem \ref{thm:main}, the requirement on the number $T$ of iterations in Theorem \ref{thm:pmgd} is pessimistic when compared to the empirical performance of our method. The experiments in Section \ref{sec:experiment} show that both MGD and PMGD achieve very high accuracy after $T=100$ iterations.

\subsection{Proof of the Algorithm Guarantees}

We first establish two lemmas on the bounds and Lipschitz continuity of the gradient and the Hessian. These properties are routinely used in the analysis of optimization methods, and will be leveraged in our proofs.

\begin{lemma}\label{lem:L_bound}
For all $h\in S^{n-1}$, $-4n^3 \leq L(h) \leq 0$, $\norm{\nabla_L(h)} \leq 16n^3$, $\norm{H_L(h)} \leq 48 n^3$.
\end{lemma}

\begin{lemma}\label{lem:lipschitz}
The Riemannian gradient and Riemannian Hessian are Lipschitz continuous:
\begin{align*}
& \norm{\hn_L(h) - \hn_L(h')} \leq 64n^3\norm{h - h'}, \\
& \norm{\hH_L(h) - \hH_L(h')} \leq 272n^3\norm{h - h'}.
\end{align*}
\end{lemma}

We first prove Theorem \ref{thm:mgd} using a recent result on why first-order methods almost always escape saddle points \cite{lee2017first}.

\begin{IEEEproof}[Proof of Theorem \ref{thm:mgd}]
We show that if the initialization $h^{(0)}$ follows a uniform distribution on $S^{n-1}$, then $h^{(t)}$ does not converge to any saddle point almost surely, by applying \cite[Theorem 2]{lee2017first}. To this end, we verify that \textbf{(C1)} the saddle points are unstable fixed points of \eqref{eq:mgd}, and that \textbf{(C2)} the differential of $\calA(\cdot)$ in \eqref{eq:mgd} is invertible.

Let $h' \coloneqq \calA(h) = P_{S^{n-1}}(h - \gamma \hn_L(h))$. The differential $D\calA(h)$ defined in \cite[Definition 4]{lee2017first} is
\begin{align}
D\calA(h) = P_{h'^\perp} P_{h^\perp}[I - \gamma \hH_L(h)]P_{h^\perp}. 
\label{eq:differential}
\end{align}
By Theorem \ref{thm:main}, all saddle points of $L(h)$ on $S^{n-1}$ are strict saddle points. At strict saddle points $\hn_L(h)=0$ and $h' = h$. Because, as we have shown, $\hH_L(h)$ has a strictly negative eigenvalue, it follows from \cite[Proposition 8]{lee2017first} that $D\calA(h)$ has at least one eigenvalue larger than $1$. Therefore, strict saddle points are unstable fixed points of \eqref{eq:mgd} (see \cite[Definition 5]{lee2017first}), i.e., \textbf{(C1)} is satisfied. 

We verify \textbf{(C2)} in the following lemma.
\begin{lemma}\label{lem:invertible_differential}
For step size $\gamma<\frac{1}{64n^3}$, and all $h\in S^{n-1}$, we have $\mathrm{det}(D\calA(h))\neq 0$.
\end{lemma}

Since conditions \textbf{(C1)} and \textbf{(C2)} are satisfied, by \cite[Theorem 2]{lee2017first}, the set of initial points that converge to saddle points have measure $0$. Therefore, a random $h^{(0)}$ uniformly distributed on $S^{n-1}$ does not converge to any saddle point almost surely.
\end{IEEEproof}

The proof of Theorem \ref{thm:pmgd} follows steps similar to those in the proof of \cite[Theorem 9]{jin2019stochastic}. The main notable difference is that we have to modify the arguments for optimization on the sphere. In particular, the following lemmas are the counter parts of \cite[Lemmas 10 -- 13]{jin2019stochastic}.

Lemma \ref{lem:gradient_descent_step} shows that every manifold gradient descent step decreases $L(h)$ by a certain amount if the Riemannian gradient is not zero.
\begin{lemma} \label{lem:gradient_descent_step}
For the manifold gradient update in \eqref{eq:mgd} with a step size $\gamma = \frac{1}{128n^3}$,
\[
L(h^{(t+1)}) - L(h^{(t)}) \leq -\frac{0.0038}{n^3} \norm{\hn_L(h^{(t)})}^2.
\]
\end{lemma}

A simple extension of the above lemma shows that, if the objective function does not decrease much over $\calT$ iterations, then every iterate is within a small neighborhood of the starting point:
\begin{lemma} \label{lem:improve_or_localize}
Give the same setup as in Lemma \ref{lem:gradient_descent_step}, for every $0\leq t\leq \calT$,
\[
\norm{h^{(0)} - h^{(t)}} \leq \sqrt{\frac{\calT}{62n^3}[L(h^{(0)}) - L(h^{(\calT)})]}.
\]
\end{lemma}

Next, we give a key result showing that regions near saddle points, where manifold gradient descent tends to get stuck, are very narrow. We prove this using the ``coupling sequence'' technique first proposed by Jin et al. \cite{jin2019stochastic}. We extend this technique to manifold gradient descent on the sphere in the following lemma. 
\begin{lemma} \label{lem:coupling_sequence}
Assume that $h_0 \in \calH_2$ and $\norm{\hn_L(h_0)} < c(n,\theta,\rho)$. Let $640 < \xi < 128/\rho^2$ denote a constant. Let $z_0 \coloneqq \arg\min_{\substack{z: \norm{z}=1\\z\perp h_0}} z^\T \hH_{L}(h_0) z$ denote the negative curvature direction.
Assume that $h^{(0)}_1$ and $h^{(0)}_2$ are two perturbed versions of $h_0$, and satisfy: (1) $\max\{\norm{h^{(0)}_1-h_0},\,\norm{h^{(0)}_2-h_0} \}\leq \frac{\calR}{2} \coloneqq \frac{\theta(1-3\theta)}{2\xi n^3\log n}$; (2) $P_{h_0^\perp} (h^{(0)}_1 - h^{(0)}_2) = d z_0 \coloneqq \frac{4\theta(1-3\theta)}{\xi n^{(3+\xi/10000)}\log n} z_0$. Then the two coupling sequences of $\calT \coloneqq \frac{\xi n^3\log n}{25\theta(1-3\theta)}$ manifold gradient descent iterates $\{h^{(t)}_1\}_{t=0}^\calT$ and $\{h^{(t)}_2\}_{t=0}^\calT$ satisfy:
\[
\max\{ L(h^{(0)}_1) - L(h^{(\calT)}_1),\,L(h^{(0)}_2) - L(h^{(\calT)}_2) \} > \calL \coloneqq \frac{256\theta^3(1-3\theta)^3}{\xi^3n^6\log^3n}.
\]
\end{lemma}

Lemma \ref{lem:coupling_sequence} effectively means that perturbation can help manifold gradient descent to escape bad regions near saddle points, which is formalized as follows:
\begin{lemma} \label{lem:perturbation_escapes_saddle}
Assume that $h_0 \in \calH_2$ and $\norm{\hn_L(h_0)} < c(n,\theta,\rho)$. Let $h^{(0)} = \sqrt{1-\norm{z_p}^2} h_0 + z_p$ be a perturbed version of $h_0$, where $z_p$ is a random vector uniformly distributed in the ball of radius $\calD \coloneqq \calR^2$ in the tangent space at $h_0$. Then after $\calT$ iterations of manifold gradient descent, we have $L(h^{\calT}) - L(h_0) < -\calL / 2$ with probability at least $1 - \frac{4\xi n^{(4-\xi/10000)}\log n}{\theta(1-3\theta)}$.
\end{lemma}

\begin{IEEEproof}[Proof of Theorem \ref{thm:pmgd}]
First, we show that among $T$ iterations there are at most 
\[
T_1 = \frac{L(h^{(0)}) - \inf L(h)}{\frac{0.0038}{n^3} c(n,\theta, \rho)^2}
\]
large gradient steps where $\norm{\hn_L(h^{(t)})} \geq c(n,\theta, \rho)$. Otherwise, one can apply Lemma \ref{lem:gradient_descent_step} to each such step, and have $L(h^{(0)}) - L(h^{(T)}) > L(h^{(0)}) - \inf L(h)$, which cannot happen.

Similarly, there are at most 
\[
T_2 = \frac{L(h^{(0)}) - \inf L(h)}{\calL / 2}
\]
perturbation steps with probability at least $1 - T_2 \frac{4\xi n^{(4-\xi/10000)}\log n}{\theta(1-3\theta)}$. Otherwise, one can apply Lemma \ref{lem:perturbation_escapes_saddle} to the perturbations, and have $L(h^{(0)}) - L(h^{(T)}) > L(h^{(0)}) - \inf L(h)$, which is again self-contradictory.

By Lemma \ref{lem:L_bound}, $L(h^{(0)}) - \inf L(h) \leq 4n^3$. Therefore,
\[
T_1 < \frac{5000n^8}{\theta^2(1-3\theta)^2\rho^4},
\] 
\[
T_2 < \frac{\xi^3n^9\log^3n}{32\theta^3(1-3\theta)^3}.
\]
It follows that, with probability at least $1 - \frac{\xi^4 n^{(13-\xi/10000)}\log^4 n}{8\theta^4(1-3\theta)^4}$, the iterates of perturbed manifold gradient descent can stay in $\calH_2 \cup \calH_3$ (large gradient or negative curvature regions) for at most $T_1 + T_2 \calT$ steps, thus completing the proof.
\end{IEEEproof}


\section{Extensions} \label{sec:extension}

We believe that our formulation and/or analysis can be extended to other scenarios that are not covered by our theoretical guarantees.
\begin{itemize}
	\itemsep0em
	\item[$\circ$] Bernoulli-subgaussian channels. As stated at the beginning of Section \ref{sec:geometry}, the Bernoulli-Rademacher assumption (A1) is a special case of the Bernoulli-subgaussian distribution, which simplifies our analysis. We conjecture that similar bounds can be established for general subgaussian distributions.
	
	\item[$\circ$] Jointly sparse channels. This is a special case where the supports of $x_i$ ($i=1,2,\dots, N$) are identical. Due to the shared support, the $x_i$'s are no longer independent. In this case, one needs a more careful analysis conditioned on the joint support.
	
	\item[$\circ$] Complex signal and channels. We mainly consider real signals in this paper. However, a similar approach can be derived and analyzed for complex signals. We discuss this extension in the rest of this section.
\end{itemize}
Empirical evidence that our method works in these scenarios is provided in Section \ref{sec:calibration}.

For complex $f, x_i\in\bbC^n$, one can solve the following problem:
\begin{align*}
\min_{h\in\bbC^n} ~~ \frac{1}{N} \sum_{i=1}^N \phi(\re(C_{y_i} R h)) + \phi(\im(C_{y_i} R h)), \quad \text{s.t.} ~~ \norm{h} = 1,
\end{align*}
where $R \coloneqq (\frac{1}{\theta n N}\sum_{i=1}^N C_{y_i}^\Herm C_{y_i})^{-1/2}\in\bbC^{n\times n}$, and $(\cdot)^\Herm$ represents the Hermitian transpose.
If one treats the real and imaginary parts of $h$ separately, then this optimization in $\bbC^n$ can be recast into $\bbR^{2n}$, and the gradient with respect to $\re(h)$ and $\im(h)$ can be used in first-order methods. This is related to Wirtinger gradient descent algorithms (see the discussion in \cite{candes2015phase}). The Riemannian gradient with respect to $h$ is:
\[
P_{(\bbR\cdot h)^\perp} \Bigl(\frac{1}{N} \sum_{i=1}^N R^\Herm C_{y_i}^\Herm w_i(h)\Bigr),
\]
where $w_i(h)$ represents the following complex vector:
\[
w_i(h) = \nabla_\phi(\re(C_{y_i}Rh)) + \sqrt{-1} \nabla_\phi(\im(C_{y_i}Rh)),
\]
and $P_{(\bbR\cdot h)^\perp}$ represents the projection onto the tangent space at $h$ in $S^{2n-1} \subset \bbR^{2n}$:
\[
P_{(\bbR\cdot h)^\perp} z = z - \re(h^\Herm z) \cdot h.
\]
In the complex case, one can initialize the manifold gradient descent algorithm with a random $h^{(0)}$, for which $[\re(h^{(0)})^\T, \im(h^{(0)})^\T]^\T$ follows a uniform distribution on $S^{2n-1}$.


\section{Numerical Experiments} \label{sec:experiment}

\subsection{Manifold Gradient Descent vs. Perturbed Manifold Gradient Descent} \label{sec:compare_perturbation}

Our theoretical analysis in Section \ref{sec:optimization} suggests that manifold gradient descent (MGD) is guaranteed to escape saddle points almost surely, but may take very long to do so. On the other hand, \emph{perturbed} manifold gradient descent (PMGD) is guaranteed to escape saddle points \emph{efficiently} with high probability. 
In this section, we compare the empirical performances of these two algorithms in solving the multichannel sparse blind deconvolution problem (P1). In particular, we investigate whether the random perturbation can help achieve faster convergence.

We synthesize $\{x_i\}_{i=1}^N$ following the Bernoulli-Rademacher model, and synthesize $f$ following a Gaussian distribution $N(\bfz_{n\times 1}, I_{n})$. Recall that the desired $h$ is a signed shifted version of the ground truth, i.e., $C_f R h = \pm e_j$ ($j\in [n]$) is a standard basis vector. Therefore, to evaluate the accuracy of the output $h^{(T)}$, we compute $C_f R h^{(T)}$ with the true $f$, and declare successful recovery if
\[
\frac{\norm{C_f R h^{(T)}}_\infty}{\norm{C_f R h^{(T)}}} > 0.95,
\]
or equivalently, if 
\[ 
\max_{j \in [n]} \bigl|\cos \angle \bigl(C_f R h^{(T)}, e_j \bigr) \bigr| > 0.95.
\]
We compute the success rate based on $100$ Monte Carlo instances.

We fix $N = 256$, $n = 128$, and $\theta=0.1$ in this section.
We run MGD and PMGD for $T = 100$ iterations, with a fixed step size of $\gamma=0.1$. We choose the perturbation time interval $\calT = 10$, and compare the following choices of perturbation radius $\calD$ and tolerance $c$:
\begin{enumerate}
\itemsep0em
	\item $\calD = 0.04$, $c = 0.1$
	\item $\calD = 0.08$, $c = 0.2$
	\item $\calD = 0.2$, $c = 0.5$
	\item $\calD = 0.4$, $c = 1$
\end{enumerate}
The empirical success rates (at iteration number $10$, $20$, $40$, $80$) of MGD and PMGD (with the above parameter choices) are shown in Figure \ref{fig:compare_perturbation}. 

Despite the clear advantage of PMGD over MGD in theoretical guarantees, we don't observe any significant disparity in empirical performances between the two. 
We conjecture that, because the ``bad regions'' near the saddle points of our objective function are extremely thin, the probability that MGD iterates get stuck at those regions is negligible. 
Therefore, we focus on assessing the performance of MGD (which is a special case of PMGD with $\calD = 0$ or $c = 0$) in the rest of the experiments.

\begin{figure}[htbp]%
\centering
\includegraphics[width=0.7\columnwidth]{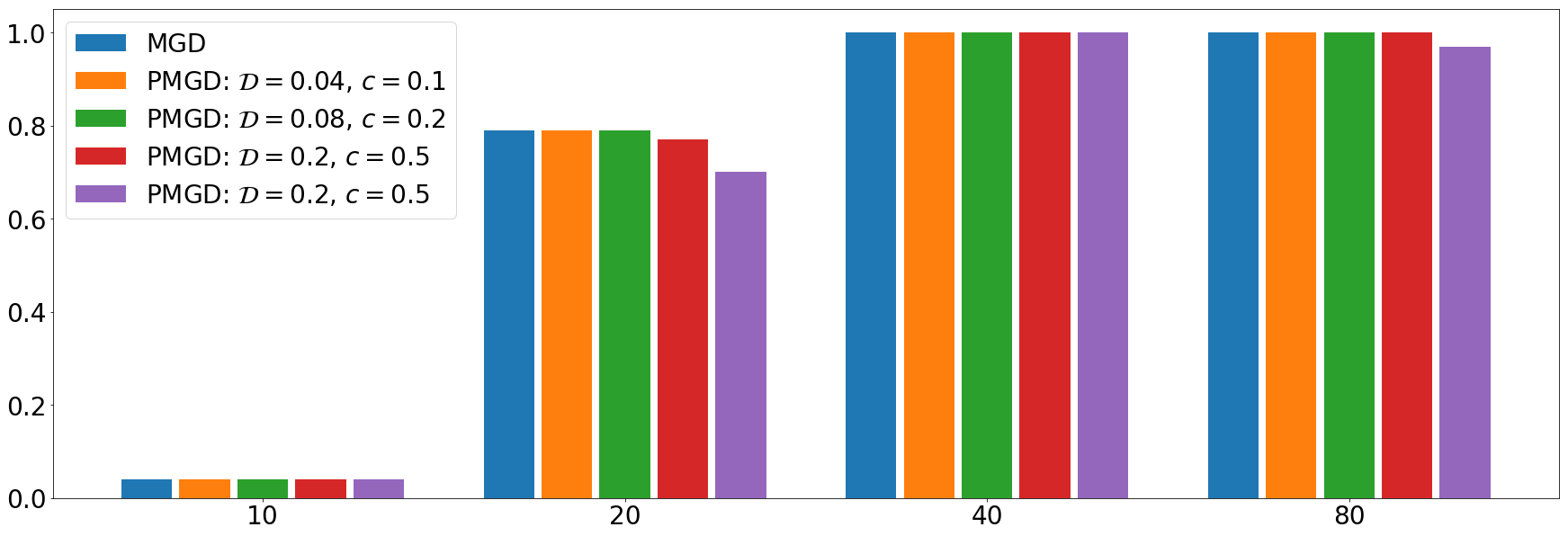}%
\caption{Empirical success rates (at iteration number $10$, $20$, $40$, $80$) of MGD and PMGD (with different parameter choices). The $x$-axis represents iteration numbers, and the $y$-axis represents success rates.}%
\label{fig:compare_perturbation}%
\end{figure}

\subsection{Deconvolution with Synthetic Data} \label{sec:synthetic}
In this section, we examine the empirical sampling complexity requirement of MGD \eqref{eq:mgd} to solve the MSBD formulation (P1). In particular, we answer the following question: how many channels $N$ are needed for different problem sizes $n$, sparsity levels $\theta$, and condition number $\kappa$? 

We follow the same experimental setup as in Section \ref{sec:compare_perturbation}. We run MGD for $T = 100$ iterations, with a fixed step size of $\gamma=0.1$.

In the first experiment, we fix $\theta = 0.1$ (sparsity level, mean of the Bernoulli distribution), and run experiments with $n = 32, 64, \dots, 256$ and $N = 32, 64, \dots, 256$ (see Figure \ref{fig:Nvsn}). In the second experiment, we fix $n=256$, and run experiments with $\theta=0.02, 0.04, \dots, 0.16$ and $N = 32, 64, \dots, 256$ (see Figure \ref{fig:Nvstheta}). The empirical phase transitions suggest that, for sparsity level relatively small (e.g., $\theta < 0.12$), there exist a constant $c>0$ such that MGD can recover a signed shifted version of the ground truth with $N\geq c n\theta$. 

In the third experiment, we examine the phase transition with respect to $N$ and the condition number $\kappa$ of $f$, which is the ratio of the largest and smallest magnitudes of its DFT. To synthesize $f$ with specific $\kappa$, we generate the DFT $\tilde{f}$ of $f$ that is random with the following distribution: 1) The DFT $\tilde{f}$ is symmetric, i.e., $\tilde{f}_{(j)} = \tilde{f}_{(n+2-j)}$, so that $f$ is real. 2) The phase of $\tilde{f}_{(j)}$ follows a uniform distribution on $[0, 2\pi)$, except for the phases of $\tilde{f}_{(1)}$ and $\tilde{f}_{(n/2+1)}$ (if $n$ is even), which are always $0$, for symmetry. 3) The gains of $\tilde{f}$ follows a uniform distribution on $[1, \kappa]$. We fix $n=256$ and $\theta=0.1$, and run experiments with $\kappa=1, 2, 4, \dots, 128$ and $N = 32, 64, \dots, 256$ (see Figure \ref{fig:Nvskappa}). The phase transition suggests that the number $N$ for successful empirical recovery is not sensitive to the condition number $\kappa$. 

\begin{figure}[htbp]%
\centering
\subfigure[]{
\begin{tikzpicture}[scale=0.6]

\begin{axis}[
width=2.5in,
height=2.5in,
xmin=-0.5, xmax=7.5,
ymin=-0.5, ymax=7.5,
tick align=outside,
tick pos=left,
xtick={1,3,5,7},
ytick={1,3,5,7},
xticklabels={64,128,192,256},
yticklabels={64,128,192,256},
xlabel={$n$},
ylabel={$N$},
x grid style={white!69.01960784313725!black},
y grid style={white!69.01960784313725!black}
]
\addplot graphics [includegraphics cmd=\pgfimage,xmin=-0.5, xmax=7.5, ymin=7.5, ymax=-0.5] {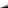};
\end{axis}

\end{tikzpicture}
	\label{fig:Nvsn}
}~
\subfigure[]{
\begin{tikzpicture}[scale=0.6]

\begin{axis}[
width=2.5in,
height=2.5in,
xmin=-0.5, xmax=7.5,
ymin=-0.5, ymax=7.5,
tick align=outside,
tick pos=left,
xtick={1,3,5,7},
ytick={1,3,5,7},
xticklabels={64,128,192,256},
yticklabels={64,128,192,256},
xlabel={$n$},
ylabel={$N$},
x grid style={white!69.01960784313725!black},
y grid style={white!69.01960784313725!black}
]
\addplot graphics [includegraphics cmd=\pgfimage,xmin=-0.5, xmax=7.5, ymin=7.5, ymax=-0.5] {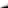};
\end{axis}

\end{tikzpicture}
	\label{fig:Nvsn_nn}
}~
\subfigure[]{
\begin{tikzpicture}[scale=0.6]

\begin{axis}[
width=2.5in,
height=2.5in,
xmin=-0.5, xmax=7.5,
ymin=-0.5, ymax=7.5,
tick align=outside,
tick pos=left,
xtick={1,3,5,7},
ytick={1,3,5,7},
xticklabels={64,128,192,256},
yticklabels={64,128,192,256},
xlabel={$n$},
ylabel={$N$},
x grid style={white!69.01960784313725!black},
y grid style={white!69.01960784313725!black}
]
\addplot graphics [includegraphics cmd=\pgfimage,xmin=-0.5, xmax=7.5, ymin=7.5, ymax=-0.5] {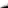};
\end{axis}

\end{tikzpicture}
	\label{fig:Nvsn_n}
}
\\
\subfigure[]{
\begin{tikzpicture}[scale=0.6]

\begin{axis}[
width=2.5in,
height=2.5in,
xmin=-0.5, xmax=7.5,
ymin=-0.5, ymax=7.5,
tick align=outside,
tick pos=left,
xtick={1,3,5,7},
ytick={1,3,5,7},
xticklabels={0.04,0.08,0.12,0.16},
yticklabels={64,128,192,256},
xlabel={$\theta$},
ylabel={$N$},
x grid style={white!69.01960784313725!black},
y grid style={white!69.01960784313725!black}
]
\addplot graphics [includegraphics cmd=\pgfimage,xmin=-0.5, xmax=7.5, ymin=7.5, ymax=-0.5] {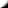};
\end{axis}

\end{tikzpicture}
	\label{fig:Nvstheta}
}~
\subfigure[]{
\begin{tikzpicture}[scale=0.6]

\begin{axis}[
width=2.5in,
height=2.5in,
xmin=-0.5, xmax=7.5,
ymin=-0.5, ymax=7.5,
tick align=outside,
tick pos=left,
xtick={1,3,5,7},
ytick={1,3,5,7},
xticklabels={0.04,0.08,0.12,0.16},
yticklabels={64,128,192,256},
xlabel={$\theta$},
ylabel={$N$},
x grid style={white!69.01960784313725!black},
y grid style={white!69.01960784313725!black}
]
\addplot graphics [includegraphics cmd=\pgfimage,xmin=-0.5, xmax=7.5, ymin=7.5, ymax=-0.5] {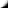};
\end{axis}

\end{tikzpicture}
	\label{fig:Nvstheta_nn}
}~
\subfigure[]{
\begin{tikzpicture}[scale=0.6]

\begin{axis}[
width=2.5in,
height=2.5in,
xmin=-0.5, xmax=7.5,
ymin=-0.5, ymax=7.5,
tick align=outside,
tick pos=left,
xtick={1,3,5,7},
ytick={1,3,5,7},
xticklabels={0.04,0.08,0.12,0.16},
yticklabels={64,128,192,256},
xlabel={$\theta$},
ylabel={$N$},
x grid style={white!69.01960784313725!black},
y grid style={white!69.01960784313725!black}
]
\addplot graphics [includegraphics cmd=\pgfimage,xmin=-0.5, xmax=7.5, ymin=7.5, ymax=-0.5] {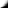};
\end{axis}

\end{tikzpicture}
	\label{fig:Nvstheta_n}
}\\
\subfigure[]{
\begin{tikzpicture}[scale=0.6]

\begin{axis}[
width=2.5in,
height=2.5in,
xmin=-0.5, xmax=7.5,
ymin=-0.5, ymax=7.5,
tick align=outside,
tick pos=left,
xtick={1,3,5,7},
ytick={1,3,5,7},
xticklabels={2,8,32,128},
yticklabels={64,128,192,256},
xlabel={$\kappa$},
ylabel={$N$},
x grid style={white!69.01960784313725!black},
y grid style={white!69.01960784313725!black}
]
\addplot graphics [includegraphics cmd=\pgfimage,xmin=-0.5, xmax=7.5, ymin=7.5, ymax=-0.5] {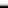};
\end{axis}

\end{tikzpicture}
	\label{fig:Nvskappa}
}~
\subfigure[]{
\begin{tikzpicture}[scale=0.6]

\begin{axis}[
width=2.5in,
height=2.5in,
xmin=-0.5, xmax=7.5,
ymin=-0.5, ymax=7.5,
tick align=outside,
tick pos=left,
xtick={1,3,5,7},
ytick={1,3,5,7},
xticklabels={2,8,32,128},
yticklabels={64,128,192,256},
xlabel={$\kappa$},
ylabel={$N$},
x grid style={white!69.01960784313725!black},
y grid style={white!69.01960784313725!black}
]
\addplot graphics [includegraphics cmd=\pgfimage,xmin=-0.5, xmax=7.5, ymin=7.5, ymax=-0.5] {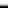};
\end{axis}

\end{tikzpicture}
	\label{fig:Nvskappa_nn}
}~
\subfigure[]{
\begin{tikzpicture}[scale=0.6]

\begin{axis}[
width=2.5in,
height=2.5in,
xmin=-0.5, xmax=7.5,
ymin=-0.5, ymax=7.5,
tick align=outside,
tick pos=left,
xtick={1,3,5,7},
ytick={1,3,5,7},
xticklabels={2,8,32,128},
yticklabels={64,128,192,256},
xlabel={$\kappa$},
ylabel={$N$},
x grid style={white!69.01960784313725!black},
y grid style={white!69.01960784313725!black}
]
\addplot graphics [includegraphics cmd=\pgfimage,xmin=-0.5, xmax=7.5, ymin=7.5, ymax=-0.5] {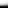};
\end{axis}

\end{tikzpicture}
	\label{fig:Nvskappa_n}
}
\caption{Empirical phase transition (grayscale values represent success rates). The first row shows the phase transitions of $N$ versus $n$, given that $\theta=0.1$. The second row shows the phase transitions of $N$ versus $\theta$, given that $n=256$. The third row shows the phase transitions of $N$ versus $\kappa$, given that $n=256$ and $\theta=0.1$. 
The first column shows the results for the noiseless case. The second column shows the results for $\mathrm{SNR}\approx 40\,\mathrm{dB}$. The third column shows are the results for $\mathrm{SNR}\approx 20\,\mathrm{dB}$. }%
\label{fig:phase_transition}%
\end{figure}

MGD is robust against noise. We repeat the above experiments with noisy measurements: $y_i = x_i \circledast f + \sigma \varepsilon_i$, where $\varepsilon_i$ follows a Gaussian distribution $N(\bfz_{n\times 1}, I_{n})$. 
The phase transitions for $\sigma = 0.01\sqrt{n\theta}$ ($\mathrm{SNR}\approx 40\,\mathrm{dB}$) and $\sigma = 0.1\sqrt{n\theta}$ ($\mathrm{SNR}\approx 20\,\mathrm{dB}$) are shown in Figure \ref{fig:Nvsn_nn}, \ref{fig:Nvstheta_nn}, \ref{fig:Nvskappa_nn}, and Figure \ref{fig:Nvsn_n}, \ref{fig:Nvstheta_n}, \ref{fig:Nvskappa_n}, respectively. For reasonable noise levels, the number $N$ of noisy measurements we need to accurately recover a signed shifted version of the ground truth is roughly the same as with noiseless measurements.


\subsection{2D Deconvolution}

Next, we run a numerical experiment with blind image deconvolution. Suppose the circular convolutions $\{y_i\}_{i=1}^N$ (Figure \ref{fig:img_y}) of an unknown image $f$ (Figure \ref{fig:img_f}) and unknown sparse channels $\{x_i\}_{i=1}^N$ (Figure \ref{fig:img_x}) are observed. The recovered image $\hat{f}$ (Figure \ref{fig:img_fhat}) is computed as follows:
\[
\hat{f} = \calF^{-1}\bigl[\calF(Rh^{(T)})^{\odot -1}\bigr],
\]
where $\calF$ denotes the 2D DFT, and $h^{(T)}$ is the output of MGD \eqref{eq:mgd}, with a random initialization $h^{(0)}$ that is uniformly distributed on the sphere.

Figure \ref{fig:img} shows that, although the sparse channels are completely unknown and the convolutional observations have corrupted the image beyond recognition, MGD is capable of recovering a shifted version of the (negative) image, starting from a random point on the sphere (see the image recovered using a random initialization in Figure \ref{fig:img_fhat}, and then corrected with the true sign and shift in Figure \ref{fig:img_correct}). In this example, all images and channels are of size $64\times 64$, the number of channels is $N=256$, and the sparsity level is $\theta=0.01$. We run $T=100$ iterations of MGD with a fixed step size $\gamma = 0.05$. The accuracy $\frac{\norm{C_f R h^{(t)}}_\infty}{\norm{C_f R h^{(t)}}}$ as a function of iteration number $t$ is shown in Figure \ref{fig:img_accuracy}, and exhibits a sharp transition at a modest number ($\approx 80$) of iterations.

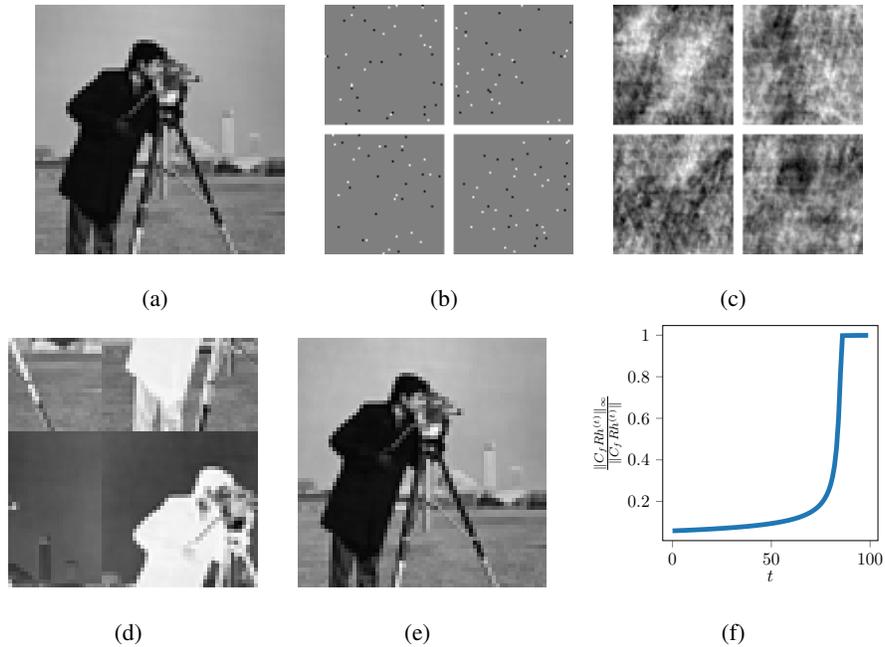
\begin{figure}[htbp]%
\centering
\subfigure[]{
	\includegraphics[width=0.2\columnwidth]{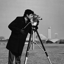}%
	\label{fig:img_f}
}~
\subfigure[]{
	\includegraphics[width=0.2\columnwidth]{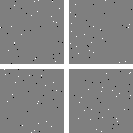}%
	\label{fig:img_x}
}~
\subfigure[]{
	\includegraphics[width=0.2\columnwidth]{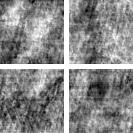}%
	\label{fig:img_y}
}\\
\subfigure[]{
	\includegraphics[width=0.2\columnwidth]{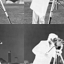}%
	\label{fig:img_fhat}
}~
\subfigure[]{
	\includegraphics[width=0.2\columnwidth]{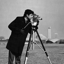}%
	\label{fig:img_correct}
}~
\subfigure[]{
\begin{tikzpicture}[scale=0.6]

\definecolor{color0}{rgb}{0.12156862745098,0.466666666666667,0.705882352941177}

\begin{axis}[
width=2.5in,
height=2.5in,
xmin=-4.95, xmax=103.95,
ymin=0.0113384346819374, ymax=1.04648476573765,
tick align=outside,
tick pos=left,
xlabel={$t$},
ylabel={$\frac{\norm{C_f R h^{(t)}}_\infty}{\norm{C_f R h^{(t)}}}$},
x grid style={white!69.01960784313725!black},
y grid style={white!69.01960784313725!black}
]
\addplot [semithick, color0, forget plot, line width=3pt]
table {%
0 0.0583905406390151
1 0.0587604663863387
2 0.0591372908387506
3 0.059521213735704
4 0.05991244369575
5 0.060311198737924
6 0.0607177068403895
7 0.0611322065395034
8 0.0615549475727849
9 0.0619861915696189
10 0.0624262127939153
11 0.0628752989433833
12 0.0633337520105662
13 0.063801889211331
14 0.0642800439871188
15 0.0647685670879531
16 0.065267827743978
17 0.0657782149341733
18 0.0663001387618822
19 0.0668340319479035
20 0.0673803514531669
21 0.0679395802444492
22 0.0685122292182208
23 0.0690988392995802
24 0.0696999837353585
25 0.0703162706029132
26 0.0709483455589224
27 0.0715968948557011
28 0.0722626486562573
29 0.0729463846835752
30 0.0736489322445527
31 0.074371176674755
32 0.0751140642568134
33 0.0758786076730791
34 0.0766658920622429
35 0.0774770817603043
36 0.078313427818839
37 0.0791762764083356
38 0.0800670782319305
39 0.0809873990957277
40 0.0819389318067569
41 0.0829235095993824
42 0.0839431213267145
43 0.0849999286966782
44 0.0860962858845819
45 0.087234761917496
46 0.0884181663032788
47 0.0896495784722428
48 0.0909323817168151
49 0.0922703024600544
50 0.0936674558652687
51 0.0951283990263765
52 0.0966581932654586
53 0.0982624774279804
54 0.0999475545313159
55 0.101720494720774
56 0.103589258263365
57 0.105562843323681
58 0.107651464602716
59 0.109866770697544
60 0.112222110425403
61 0.114732861591079
62 0.117416840112273
63 0.120294813570885
64 0.123391151901604
65 0.126734660236425
66 0.130359656705904
67 0.134307384101661
68 0.138627883290402
69 0.143382515600167
70 0.148647413568588
71 0.154518285899095
72 0.161117241059096
73 0.168602693495738
74 0.177184106325362
75 0.187144556918401
76 0.198876399592291
77 0.212939731231072
78 0.230162368559797
79 0.251819320249649
80 0.279973365469553
81 0.318162442434067
82 0.372874510226094
83 0.456810343391301
84 0.594966359033851
85 0.816938146074972
86 0.997989926648266
87 0.998286372252348
88 0.99851830768603
89 0.998701791841688
90 0.99884984017516
91 0.998969817760895
92 0.999068844073349
93 0.999150518577022
94 0.999219144335579
95 0.999276503166327
96 0.99932541001094
97 0.999366701286928
98 0.999402345243989
99 0.999432659780569
};
\end{axis}

\end{tikzpicture}%
	\label{fig:img_accuracy}
}
\caption{Multichannel blind image deconvolution. (a) True image. (b) Sparse channels. (c) Observations. (d) Recovered image using MGD. (e) Recovered image with sign and shift correction. (f) The accuracy as a function of iteration number. All images and channels in this figure are of the same size ($64 \times 64$).}%
\label{fig:img}%
\end{figure}


\subsection{Jointly Sparse Complex Gaussian Channels} \label{sec:calibration}

In this section, we examine the performance of MGD when Assumption (A1) is not satisfied, and the channel model is extended as in Section \ref{sec:extension}. More specifically, we consider $f$ following a Gaussian distribution $CN(\bfz_{n\times 1}, I_{n})$, i.e., the real and imaginary parts are independent following $N(\bfz_{n\times 1}, I_{n} / 2)$. And we consider $\{x_i\}_{i=1}^N$ that are:
\begin{itemize}
	\itemsep0em
	\item[$\circ$] \textbf{Jointly $s$-sparse:} The joint support of $\{x_i\}_{i=1}^N$ is chosen uniformly at random on $[n]$.
	\item[$\circ$] \textbf{Complex Gaussian:} The nonzero entries of $\{x_i\}_{i=1}^N$ follow a complex Gaussian distribution $CN(0, 1)$.
\end{itemize}

We compare MGD (with random initialization) with three blind calibration algorithms that solve MSBD in the frequency domain: truncated power iteration \cite{li2017blind} (initialized with $f^{(0)}=e_1$ and $x_i^{(0)}=0$), an off-the-grid algebraic method \cite{wylie1993self} (simplified from \cite{paulraj1985direction}), and an off-the-grid optimization approach \cite{eldar2018sensor}. 

We fix $n=128$, and run experiments for $N=16, 32, 48, \cdots, 128$, and $s = 2,4,6,\dots,16$. We use $f$ and $\hat{f}$ to denote the true signal and the recovered signal, respectively. We say the recovery is successful if\footnote{A perfect recovery $\hat{f}$ is a scaled shifted version of $f$, for which $\calF^{-1}\bigl[\calF(f) \odot \calF(\hat{f})^{\odot -1}\bigr]$ is a scaled shifted Kronecker delta.}
\begin{align}
\frac{ \norm{\calF^{-1}\bigl[\calF(f) \odot \calF(\hat{f})^{\odot -1}\bigr]}_\infty }{ \norm{\calF^{-1}\bigl[\calF(f) \odot \calF(\hat{f})^{\odot -1}\bigr]} } > 0.7. \label{eq:success}
\end{align}

By the phase transitions in Figure \ref{fig:alg_compare}, MGD and truncated power iteration are successful when $N$ is large and $s$ is small. Although truncated power iteration achieves higher success rates when both $N$ and $s$ are small, it fails for $s>8$ even with a large $N$. On the other hand, MGD can recover channels with $s=16$ when $N\geq 80$.\footnote{By our theoretical prediction, MGD can succeed for $s = 40 < \frac{n}{3}$ provided that we have a sufficiently large $N$.} In comparison, the off-the-grid methods are based on the properties of the covariance matrix $\frac{1}{N} \sum_{i=1}^N y_i y_i^\Herm$, and require larger $N$ (than the first two algorithms) to achieve high success rates.
 
\begin{figure}[htbp]%
\centering
\subfigure[]{
\begin{tikzpicture}[scale=0.6]

\begin{axis}[
width=2.5in,
height=2.5in,
xmin=-0.5, xmax=7.5,
ymin=-0.5, ymax=7.5,
tick align=outside,
tick pos=left,
xtick={1,3,5,7},
ytick={1,3,5,7},
xticklabels={4,8,12,16},
yticklabels={32,64,96,128},
xlabel={$s$},
ylabel={$N$},
x grid style={white!69.01960784313725!black},
y grid style={white!69.01960784313725!black}
]
\addplot graphics [includegraphics cmd=\pgfimage,xmin=-0.5, xmax=7.5, ymin=7.5, ymax=-0.5] {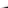};
\end{axis}

\end{tikzpicture}
	\label{fig:alg_compare_1}
}~
\subfigure[]{
\begin{tikzpicture}[scale=0.6]

\begin{axis}[
width=2.5in,
height=2.5in,
xmin=-0.5, xmax=7.5,
ymin=-0.5, ymax=7.5,
tick align=outside,
tick pos=left,
xtick={1,3,5,7},
ytick={1,3,5,7},
xticklabels={4,8,12,16},
yticklabels={32,64,96,128},
xlabel={$s$},
ylabel={$N$},
x grid style={white!69.01960784313725!black},
y grid style={white!69.01960784313725!black}
]
\addplot graphics [includegraphics cmd=\pgfimage,xmin=-0.5, xmax=7.5, ymin=7.5, ymax=-0.5] {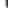};
\end{axis}

\end{tikzpicture}
	\label{fig:alg_compare_2}
}\\
\subfigure[]{
\begin{tikzpicture}[scale=0.6]

\begin{axis}[
width=2.5in,
height=2.5in,
xmin=-0.5, xmax=7.5,
ymin=-0.5, ymax=7.5,
tick align=outside,
tick pos=left,
xtick={1,3,5,7},
ytick={1,3,5,7},
xticklabels={4,8,12,16},
yticklabels={32,64,96,128},
xlabel={$s$},
ylabel={$N$},
x grid style={white!69.01960784313725!black},
y grid style={white!69.01960784313725!black}
]
\addplot graphics [includegraphics cmd=\pgfimage,xmin=-0.5, xmax=7.5, ymin=7.5, ymax=-0.5] {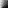};
\end{axis}

\end{tikzpicture}
	\label{fig:alg_compare_3}
}~
\subfigure[]{
\begin{tikzpicture}[scale=0.6]

\begin{axis}[
width=2.5in,
height=2.5in,
xmin=-0.5, xmax=7.5,
ymin=-0.5, ymax=7.5,
tick align=outside,
tick pos=left,
xtick={1,3,5,7},
ytick={1,3,5,7},
xticklabels={4,8,12,16},
yticklabels={32,64,96,128},
xlabel={$s$},
ylabel={$N$},
x grid style={white!69.01960784313725!black},
y grid style={white!69.01960784313725!black}
]
\addplot graphics [includegraphics cmd=\pgfimage,xmin=-0.5, xmax=7.5, ymin=7.5, ymax=-0.5] {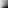};
\end{axis}

\end{tikzpicture}
	\label{fig:alg_compare_4}
}
\caption{Empirical phase transition of $N$ versus $s$, given that $n=128$. (a) MGD. (b) Truncated power iteration \cite{li2017blind}. (c) Off-the-grid algebraic method \cite{wylie1993self}. (d) Off-the-grid optimization approach \cite{eldar2018sensor}. }%
\label{fig:alg_compare}%
\end{figure}


\subsection{MSBD with a Linear Convolution Model}

In this section, we empirically study MSBD with a linear convolution model. Suppose the observations $y_i = x_i' * f' \in\bbR^n$ ($i=1,2,\dots,N$) are linear convolutions of $s$-sparse channels $x'_i \in \bbR^m$ and a signal $f'\in \bbR^{n-m+1}$. Let $x_i\in\bbR^n$ and $f\in\bbR^n$ denote the zero-padded versions of $x_i'$ and $f$. Then
\[
y_i = x_i' * f' = x_i \circledast f.
\]
In this section, we show that one can solve for $f$ and $x_i$ using the optimization formulation (P1) and MGD, without knowledge of the length $m$ of the channels. 

We compare our approach to the subspace method based on cross convolution \cite{xu1995least}, which solves for the concatenation of the channels as a null vector of a structured matrix. For fairness, we also compare to an alternative method that takes advantage of the sparsity of the channels, and finds a sparse null vector of the same structured matrix as in \cite{xu1995least}, using truncated power iteration \cite{yuan2013truncated,li2017identifiability}.\footnote{For an example of finding sparse null vectors using truncated power iteration, we refer the readers to our previous paper \cite[Section II]{li2017identifiability}.}

In our experiments, we synthesize $f'$ using a random Gaussian vector following $N(\bfz_{(n-m+1)\times 1}, I_{n-m+1})$. We synthesize $s$-sparse channels $x_i$ such that the support is chosen uniformly at random, and the nonzero entries are independent following $N(0,1)$. We denote the zero-padded versions of the true signal and the recovered signal by $f$ and $\hat{f}$, and declare success if \eqref{eq:success} is satisfied. We study the empirical success rates of our method and the competing methods in three experiments:
\begin{itemize}
	\itemsep0em
	\item[$\circ$] $N$ versus $s$, given that $n=128$ and $m=64$.
	\item[$\circ$] $N$ versus $m$, given that $n=128$ and $s=4$.
	\item[$\circ$] $N$ versus $n$, given that $m=64$ and $s=4$.
\end{itemize}

The phase transitions in Figure \ref{fig:mvs} show that our MGD method consistently has higher success rates than the competing methods based on cross convolution. The subspace method and the truncated power iteration method are only successful when $m$ is small compared to $n$, while our method is successful for a large range of $m$ and $n$. The sparsity prior exploited by truncated power iteration improves the success rate over the subspace method, but only when the sparsity level $s$ is small compared to $m$. In comparison, our method, given a sufficiently large number $N$ of channels, can recover channels with a much larger $s$.

\begin{figure}[htbp]%
\centering
\subfigure[]{
\begin{tikzpicture}[scale=0.6]

\begin{axis}[
width=2.5in,
height=2.5in,
xmin=-0.5, xmax=7.5,
ymin=-0.5, ymax=7.5,
tick align=outside,
tick pos=left,
xtick={1,3,5,7},
ytick={1,3,5,7},
xticklabels={4,8,12,16},
yticklabels={16,32,48,64},
xlabel={$s$},
ylabel={$N$},
x grid style={lightgray!92.02614379084967!black},
y grid style={lightgray!92.02614379084967!black}
]
\addplot graphics [includegraphics cmd=\pgfimage,xmin=-0.5, xmax=7.5, ymin=7.5, ymax=-0.5] {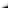};
\end{axis}

\end{tikzpicture}
	\label{fig:mvs_N_s_1}
}~
\subfigure[]{
\begin{tikzpicture}[scale=0.6]

\begin{axis}[
width=2.5in,
height=2.5in,
xmin=-0.5, xmax=7.5,
ymin=-0.5, ymax=7.5,
tick align=outside,
tick pos=left,
xtick={1,3,5,7},
ytick={1,3,5,7},
xticklabels={4,8,12,16},
yticklabels={16,32,48,64},
xlabel={$s$},
ylabel={$N$},
x grid style={lightgray!92.02614379084967!black},
y grid style={lightgray!92.02614379084967!black}
]
\addplot graphics [includegraphics cmd=\pgfimage,xmin=-0.5, xmax=7.5, ymin=7.5, ymax=-0.5] {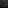};
\end{axis}

\end{tikzpicture}
	\label{fig:mvs_N_s_2}
}~
\subfigure[]{
\begin{tikzpicture}[scale=0.6]

\begin{axis}[
width=2.5in,
height=2.5in,
xmin=-0.5, xmax=7.5,
ymin=-0.5, ymax=7.5,
tick align=outside,
tick pos=left,
xtick={1,3,5,7},
ytick={1,3,5,7},
xticklabels={4,8,12,16},
yticklabels={16,32,48,64},
xlabel={$s$},
ylabel={$N$},
x grid style={lightgray!92.02614379084967!black},
y grid style={lightgray!92.02614379084967!black}
]
\addplot graphics [includegraphics cmd=\pgfimage,xmin=-0.5, xmax=7.5, ymin=7.5, ymax=-0.5] {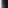};
\end{axis}

\end{tikzpicture}
	\label{fig:mvs_N_s_3}
}
\\
\subfigure[]{
\begin{tikzpicture}[scale=0.6]

\begin{axis}[
width=2.5in,
height=2.5in,
xmin=-0.5, xmax=7.5,
ymin=-0.5, ymax=7.5,
tick align=outside,
tick pos=left,
xtick={1,3,5,7},
ytick={1,3,5,7},
xticklabels={16,32,48,64},
yticklabels={16,32,48,64},
xlabel={$m$},
ylabel={$N$},
x grid style={lightgray!92.02614379084967!black},
y grid style={lightgray!92.02614379084967!black}
]
\addplot graphics [includegraphics cmd=\pgfimage,xmin=-0.5, xmax=7.5, ymin=7.5, ymax=-0.5] {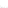};
\end{axis}

\end{tikzpicture}
	\label{fig:mvs_N_m_1}
}~
\subfigure[]{
\begin{tikzpicture}[scale=0.6]

\begin{axis}[
width=2.5in,
height=2.5in,
xmin=-0.5, xmax=7.5,
ymin=-0.5, ymax=7.5,
tick align=outside,
tick pos=left,
xtick={1,3,5,7},
ytick={1,3,5,7},
xticklabels={16,32,48,64},
yticklabels={16,32,48,64},
xlabel={$m$},
ylabel={$N$},
x grid style={lightgray!92.02614379084967!black},
y grid style={lightgray!92.02614379084967!black}
]
\addplot graphics [includegraphics cmd=\pgfimage,xmin=-0.5, xmax=7.5, ymin=7.5, ymax=-0.5] {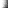};
\end{axis}

\end{tikzpicture}
	\label{fig:mvs_N_m_2}
}~
\subfigure[]{
\begin{tikzpicture}[scale=0.6]

\begin{axis}[
width=2.5in,
height=2.5in,
xmin=-0.5, xmax=7.5,
ymin=-0.5, ymax=7.5,
tick align=outside,
tick pos=left,
xtick={1,3,5,7},
ytick={1,3,5,7},
xticklabels={16,32,48,64},
yticklabels={16,32,48,64},
xlabel={$m$},
ylabel={$N$},
x grid style={lightgray!92.02614379084967!black},
y grid style={lightgray!92.02614379084967!black}
]
\addplot graphics [includegraphics cmd=\pgfimage,xmin=-0.5, xmax=7.5, ymin=7.5, ymax=-0.5] {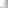};
\end{axis}

\end{tikzpicture}
	\label{fig:mvs_N_m_3}
}
\\
\subfigure[]{
\begin{tikzpicture}[scale=0.6]

\begin{axis}[
width=2.5in,
height=2.5in,
xmin=-0.5, xmax=7.5,
ymin=-0.5, ymax=7.5,
tick align=outside,
tick pos=left,
xtick={1,3,5,7},
ytick={1,3,5,7},
xticklabels={80,96,112,128},
yticklabels={16,32,48,64},
xlabel={$n$},
ylabel={$N$},
x grid style={lightgray!92.02614379084967!black},
y grid style={lightgray!92.02614379084967!black}
]
\addplot graphics [includegraphics cmd=\pgfimage,xmin=-0.5, xmax=7.5, ymin=7.5, ymax=-0.5] {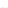};
\end{axis}

\end{tikzpicture}
	\label{fig:mvs_N_n_1}
}~
\subfigure[]{
\begin{tikzpicture}[scale=0.6]

\begin{axis}[
width=2.5in,
height=2.5in,
xmin=-0.5, xmax=7.5,
ymin=-0.5, ymax=7.5,
tick align=outside,
tick pos=left,
xtick={1,3,5,7},
ytick={1,3,5,7},
xticklabels={80,96,112,128},
yticklabels={16,32,48,64},
xlabel={$n$},
ylabel={$N$},
x grid style={lightgray!92.02614379084967!black},
y grid style={lightgray!92.02614379084967!black}
]
\addplot graphics [includegraphics cmd=\pgfimage,xmin=-0.5, xmax=7.5, ymin=7.5, ymax=-0.5] {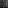};
\end{axis}

\end{tikzpicture}
	\label{fig:mvs_N_n_2}
}~
\subfigure[]{
\begin{tikzpicture}[scale=0.6]

\begin{axis}[
width=2.5in,
height=2.5in,
xmin=-0.5, xmax=7.5,
ymin=-0.5, ymax=7.5,
tick align=outside,
tick pos=left,
xtick={1,3,5,7},
ytick={1,3,5,7},
xticklabels={80,96,112,128},
yticklabels={16,32,48,64},
xlabel={$n$},
ylabel={$N$},
x grid style={lightgray!92.02614379084967!black},
y grid style={lightgray!92.02614379084967!black}
]
\addplot graphics [includegraphics cmd=\pgfimage,xmin=-0.5, xmax=7.5, ymin=7.5, ymax=-0.5] {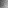};
\end{axis}

\end{tikzpicture}
	\label{fig:mvs_N_n_3}
}
\caption{Empirical phase transition of MSBD with a linear convolution model. The first row shows the phase transitions of $N$ versus $s$. The second row shows the phase transitions of $N$ versus $m$. The third row shows the phase transitions of $N$ versus $n$. 
The first column shows the results for MGD. The second column shows the results for the subspace method \cite{xu1995least}. The third column shows are the results for truncated power iteration.}%
\label{fig:mvs}%
\end{figure}


\subsection{Super-Resolution Fluorescence Microscopy} \label{sec:micro}

MGD can be applied to deconvolution of time resolved fluorescence microscopy images. The goal is to recover sharp images $x_i$'s from observations $y_i$'s that are blurred by an unknown PSF $f$. 

We use a publicly available microtubule dataset \cite{mukamel2012statistical}, which contains $N = 626$ images (Figure \ref{fig:micro_true1}). Since fluorophores are are turned on and off stochastically, the images $x_i$'s are random sparse samples of the $64\times 64$ microtubule image (Figure \ref{fig:micro_true2}). The observations $y_i$'s (Figure \ref{fig:micro_blur1}, \ref{fig:micro_blur2}) are synthesized by circular convolutions with the PSF in Figure \ref{fig:micro_k_true}. The recovered images (Figure \ref{fig:micro_recon1}, \ref{fig:micro_recon2}) and kernel (Figure \ref{fig:micro_k_recon}) clearly demonstrate the effectiveness of our approach in this setting.

Blind deconvolution is less sensitive to instrument calibration error than non-blind deconvolution. If the PSF used in a non-blind deconvolution method fails to account for certain optic aberration, the resulting images may suffer from spurious artifacts. For example, if we use a miscalibrated PSF (Figure \ref{fig:micro_k_wrong}) in non-blind image reconstruction using FISTA \cite{beck2009fast}, then the recovered images (Figure \ref{fig:micro_nonblind1}, \ref{fig:micro_nonblind2}) suffer from serious spurious artifacts.

\begin{figure}[htbp]%
\centering
\subfigure[]{
	\includegraphics[width=0.2\columnwidth]{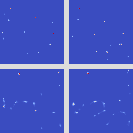}
	\label{fig:micro_true1}
}~
\subfigure[]{
	\includegraphics[width=0.2\columnwidth]{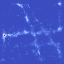}
	\label{fig:micro_true2}
}~
\subfigure[]{
	\includegraphics[width=0.2\columnwidth]{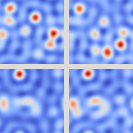}
	\label{fig:micro_blur1}
}~
\subfigure[]{
	\includegraphics[width=0.2\columnwidth]{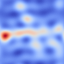}
	\label{fig:micro_blur2}
}\\
\subfigure[]{
	\includegraphics[width=0.2\columnwidth]{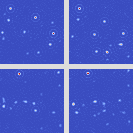}
	\label{fig:micro_recon1}
}~
\subfigure[]{
	\includegraphics[width=0.2\columnwidth]{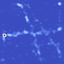}
	\label{fig:micro_recon2}
}
\subfigure[]{
	\includegraphics[width=0.2\columnwidth]{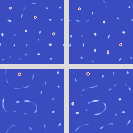}
	\label{fig:micro_nonblind1}
}~
\subfigure[]{
	\includegraphics[width=0.2\columnwidth]{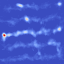}
	\label{fig:micro_nonblind2}
}\\
\subfigure[]{
	\includegraphics[width=0.3\columnwidth]{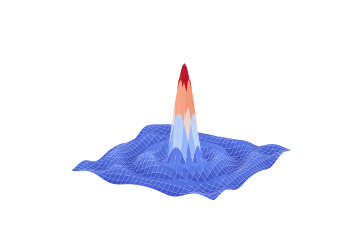}
	\label{fig:micro_k_true}
}~
\subfigure[]{
	\includegraphics[width=0.3\columnwidth]{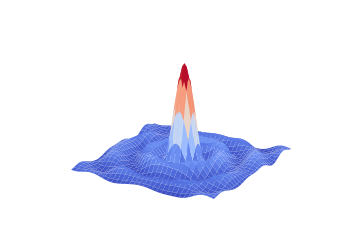}
	\label{fig:micro_k_recon}
}~
\subfigure[]{
	\includegraphics[width=0.3\columnwidth]{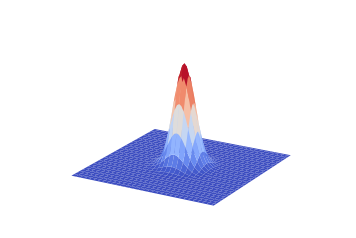}
	\label{fig:micro_k_wrong}
}
\caption{Super-resolution fluorescence microscopy experiment using MGD. (a) True images. (b) Average of true images. (c) Observed images. (d) Average of observed images. (e) Recovered images using blind deconvolution. (f) Average of recovered images using blind deconvolution. (g) Recovered images using non-blind deconvolution and a miscalibrated PSF. (h) Average of recovered images using non-blind deconvolution and a miscalibrated PSF.
(i) True PSF. (j) Recovered PSF using blind deconvolution. (k) Miscalibrated PSF used in non-blind deconvolution. All images in this figure are of the same size ($64 \times 64$).}%
\label{fig:micro}%
\end{figure}

\section{Conclusion} \label{sec:conclusion}
In this paper, we study the geometric structure of multichannel sparse blind deconvolution over the unit sphere. Our theoretical analysis reveals that local minima of a sparsity promoting smooth objective function correspond to signed shifted version of the ground truth, and saddle points have strictly negative curvatures. Thanks to the favorable geometric properties of the objective, we can simultaneously recover the unknown signal and unknown channels from convolutional measurements using (perturbed) manifold gradient descent.

In practice, many convolutional measurement models are subsampled in the spatial domain (e.g., image super-resolution) or in the frequency domain (e.g., radio astronomy). Studying the effect of subsampling on the geometric structure of multichannel sparse blind deconvolution is an interesting problem for future work.


\appendices

\section{Proofs for Section \ref{sec:geometry}} \label{app:geometry}

\subsection{Derivation of \eqref{eq:gradient} and \eqref{eq:Hessian}} \label{app:g_n_H}
Recall that
\begin{align*}
\nabla_{L''}(h) = \frac{1}{N}\sum_{i=1}^N \nabla''_i, \\
H_{L''}(h) = \frac{1}{N}\sum_{i=1}^N H''_i, 
\end{align*}
where $\nabla''_i \coloneqq C_{x_i}^\T \nabla_\phi(C_{x_i} h)$, and $H''_i = C_{x_i}^\T H_\phi(C_{x_i} h) C_{x_i}$.

For the Bernoulli-Rademacher model in (A1), we have
\begin{align*}
\bbE \nabla''_{i(j)}
& = -\bbE \sum_{s=1}^n x_{i(1+s-j)} \Bigl(\sum_{t=1}^n x_{i(1+s-t)}h_{(t)}\Bigr)^3 \\
& = -n \Bigl(\theta h_{(j)}^3 + 3 \theta^2 h_{(j)} \sum_{\ell\neq j} h_{(\ell)}^2 \Bigr) \\
& = -n \theta (1-3\theta) h_{(j)}^3 - 3 n \theta^2 h_{(j)},
\end{align*}
where the last line uses the fact that $\sum_{j=1}^n h_{(j)}^2 = \norm{h} = 1$.
Therefore, the gradient and the Riemannian gradient are
\begin{align*}
& \bbE \nabla_{L''}(h) = - n\theta(1-3\theta)h^{\odot 3} - 3n\theta^2 h,  \\
& \bbE \hn_{L''}(h) = P_{h^\perp} \bbE \nabla_{L''}(h) = n\theta(1-3\theta)(\norm{h}_4^4 \cdot h - h^{\odot 3}). 
\end{align*}

Similarly, we have
\begin{align*}
\bbE H''_{i(jk)} 
& = -3 \bbE \sum_{s=1}^n x_{i(1+s-j)}x_{i(1+s-k)} \Bigl(\sum_{t=1}^n x_{i(1+s-t)}h_{(t)}\Bigr)^2 \\
& = -3n \times \begin{cases}
\theta h_{(j)}^2 + \theta^2 \sum_{\ell\neq j} h_{(\ell)}^2 & \text{if $j = k$} \\
2\theta^2 h_{(j)}h_{(k)} & \text{if $j \neq k$}
\end{cases} \\
& = -3n \bigl[ \theta^2 \delta_{jk} + \theta(1-3\theta)h_{(j)}^2 \delta_{jk} + 2\theta^2 h_{(j)}h_{(k)} \bigr].
\end{align*}
The Hessian and the Riemannian Hessian are
\begin{align*}
& \bbE H_{L''}(h) = -3n \bigl[ \theta^2 I + \theta(1-3\theta) \diag(h^{\odot 2}) + 2\theta^2 hh^\T \bigr],  \\
& \bbE \hH_{L''}(h) = P_{h^\perp} \bbE H_{L''}(h) P_{h^\perp} - h^\T \bbE \nabla_{L''}(h) \cdot P_{h^\perp}  \\
& \qquad = n\theta(1-3\theta) \bigl[\norm{h}_4^4 \cdot I + 2\norm{h}_4^4 \cdot hh^\T - 3\cdot \diag(h^{\odot 2}) \bigr].  
\end{align*}

\subsection{Proofs of Lemmas in Section \ref{sec:geometry}}

\begin{IEEEproof}[Proof of Lemma \ref{lem:partition}]
We first investigate the Riemannian Hessian at points in $\calH''_1$ and $\calH''_2$. Without loss of generality, we consider points close to the representative stationary point $h_0 = [1/\sqrt{r},\dots,1/\sqrt{r}, 0,\dots, 0]$. We have
\begin{align*}
& |h_{(j)}^2 - 1/r| \leq \rho / r, \qquad \forall j \in \{1,2,\dots, r\}, \\
& h_{(j)}^2 \leq \rho / r, \qquad \forall j \in \{r+1, \dots, n\}, \\
& \sum_{j=r+1}^n h_{(j)}^2 = 1 - \sum_{j=1}^r h_{(j)}^2 \leq \rho.
\end{align*}
Therefore,
\begin{align}
\norm{h - h_0} \leq \sqrt{r \times \bigl(\frac{1-\sqrt{1-\rho}}{\sqrt{r}}\bigr)^2 + \rho} 
\leq \sqrt{2\rho},  \label{eq:neighbor_diff}
\end{align}
\begin{align}
\left\|\diag(h^{\odot 2}) - \frac{1}{r}
\begin{bmatrix}
I_{r} & \\
& \bfz_{(n-r)\times(n-r)}
\end{bmatrix}
\right\| \leq \frac{\rho}{r}, \label{eq:diag_diff}
\end{align}
and
\begin{align}
& \left\| hh^\T - \frac{1}{r}
\begin{bmatrix}
\bfo_{r\times r} & \\
& \bfz_{(n-r)\times(n-r)}
\end{bmatrix}
\right\| 
 \leq 2 \norm{h - h_0}
 \leq 2\sqrt{2\rho}. \label{eq:outer_diff}
\end{align}
We also bound $\norm{h}_4^4$ as follows:
\begin{align*}
\norm{h}_4^4 \leq r\times \frac{(1+\rho)^2}{r^2} + \min\Bigl\{ (n-r)\times \frac{\rho^2}{r^2},~ \frac{\rho^2}{n-r}\Bigr\} \leq \frac{1+2\rho+2\rho^2}{r},
\end{align*}
\begin{align*}
\norm{h}_4^4 \geq r\times \frac{(1-\rho)^2}{r^2} \geq \frac{1-2\rho+\rho^2}{r}.
\end{align*}
Since $\rho < 10^{-3} < 1/2$,
\begin{align}
\Bigl|\norm{h}_4^4 - \frac{1}{r} \Bigr| \leq \frac{3\rho}{r}. \label{eq:four_diff}
\end{align}

Next we obtain bounds on the Riemannian curvature of $\bbE L''$ at points $h \in \calH''_1$ or $h \in  \calH''_2$  by bounding its deviation from the Riemannian curvature at a corresponding stationary point $h_0$.
By \eqref{eq:diag_diff}, \eqref{eq:outer_diff}, \eqref{eq:four_diff}, and the expressions in \eqref{eq:Hessian}, \eqref{eq:Hessian_stationary}:
\begin{align}
& \norm{\bbE \hH_{L''}(h) - \bbE \hH_{L''}(h_0)}   \nonumber\\
& \leq n\theta(1-3\theta) \Bigl[ \frac{3\rho}{r} + 2 \times \frac{3\rho + 2\sqrt{2\rho}}{r} + 3 \times \frac{\rho}{r} \Bigr]   \nonumber\\
& = \frac{n\theta(1-3\theta)}{r} (12\rho + 4\sqrt{2\rho}).  \label{eq:Hessian_diff}
\end{align}
It follows that
\begin{align}
& \Bigl|\min_{\substack{z: \norm{z}=1\\z\perp h}} z^\T \bbE \hH_{L''}(h) z - \min_{\substack{z: \norm{z}=1\\z\perp h_0}} z^\T \bbE \hH_{L''}(h_0) z \Bigr|   \nonumber\\
& \leq \Bigl|\min_{\substack{z: \norm{z}=1\\z\perp h}} z^\T  \bbE \hH_{L''}(h) z - \min_{\substack{z: \norm{z}=1\\z\perp h}} z^\T  \bbE \hH_{L''}(h_0) z \Bigr| + \Bigl|\min_{\substack{z: \norm{z}=1\\z\perp h}} z^\T  \bbE \hH_{L''}(h_0) z - \min_{\substack{z: \norm{z}=1\\z\perp h_0}} z^\T  \bbE \hH_{L''}(h_0) z \Bigr|   \nonumber\\
& \leq \norm{V^\T\bbE \hH_{L''}(h)V - V^\T\bbE \hH_{L''}(h_0)V} + \norm{V^\T\bbE \hH_{L''}(h_0)V - V_0^\T\bbE \hH_{L''}(h_0)V_0} \nonumber\\
& \leq \norm{\bbE \hH_{L''}(h) - \bbE \hH_{L''}(h_0)} + 2 \norm{\bbE \hH_{L''}(h_0)} \cdot \norm{V- V_0} \nonumber\\
& \leq \frac{n\theta(1-3\theta)}{r} (12\rho + 4\sqrt{2\rho}) + 2 \times \frac{2n\theta(1-3\theta)}{r} \times \sqrt{2\rho}  \nonumber\\
& = \frac{n\theta(1-3\theta)}{r} (12\rho + 8\sqrt{2\rho})  \nonumber\\
& \leq \frac{n\theta(1-3\theta)(24\sqrt{\rho})}{r},  \label{eq:neighborhood}
\end{align}
where $V,\,V_0\in\bbR^{n\times (n-1)}$ satisfy: (I) the columns of $V$ (resp. $V_0$) form an orthonormal basis for the tangent space at $h$ (resp. $h_0$); (II) $\norm{V-V_0} \leq\sqrt{2\rho}$. We construct $V$ and $V_0$ as follows, for the non-trivial case where $h \neq h_0$.   
Suppose the columns of $V_\cap\in\bbR^{n\times (n-2)}$ form an orthonormal basis for the intersection of the tangent spaces at $h$ and at $h_0$. Let $c\coloneqq \langle h, h_0\rangle < 1$, and let $h' \coloneqq \frac{1}{\sqrt{1-c^2}}(h_0 - c h)$ and $h_0'\coloneqq \frac{1}{\sqrt{1-c^2}}(ch_0 - h)$. It is easy to verify that $V \coloneqq [V_\cap,\, h']$ and $V_0 \coloneqq [V_\cap,\, h_0']$ satisfy (I). To verify (II), we have $\norm{V-V_0} = \norm{h'-h_0'} = \frac{1-c}{\sqrt{1-c^2}} \norm{h+h_0} = \norm{h-h_0} \leq \sqrt{2\rho}$.

Positive definiteness \eqref{eq:PD2} follows from \eqref{eq:PD} and \eqref{eq:neighborhood}. Negative curvature \eqref{eq:NC2} follows from \eqref{eq:NC} and \eqref{eq:neighborhood}.

Next, we prove contrapositive of \eqref{eq:NG2}, i.e., suppose $\norm{\bbE \hn_{L''}(h)} < {\theta(1-3\theta)\rho^2}/{n}$ for some $h\in S^{n-1}$, then we show $h \in \calH''_1\cup \calH''_2$. First, it follows from $\norm{\bbE \hn_{L''}(h)} < {\theta(1-3\theta)\rho^2}/{n}$, and the expression in \eqref{eq:gradient_norm}, that for all $j, k\in[n]$,
\[
h_{(j)}^2h_{(k)}^2(h_{(j)}^2 - h_{(k)}^2)^2 < \frac{\rho^4}{n^4}.
\]
As a result, $|h_{(j)}^2 - h_{(k)}^2| < \rho / n$ if $h_{(j)}^2 \geq \rho / n$ and $h_{(k)}^2 \geq \rho / n$. 

Let $\Omega\coloneqq \{j: h_{(j)}^2 \geq \rho / n\} \subset[n]$, and $r\coloneqq |\Omega|$. 
Then
\begin{align}
h_{(j)}^2 < \rho/n \leq \rho / r, \qquad \forall j \notin \Omega, \label{eq:small_val}
\end{align}
and 
\begin{align}
1-(n-r)\cdot \frac{\rho}{n} < \sum_{j\in\Omega} h_{(j)}^2 \leq 1. \label{eq:sum_large}
\end{align}
In addition,
$|h_{(j)}^2 - h_{(k)}^2| < \rho / n$ for $j, k\in\Omega$. Therefore, for $k\in \Omega$, $h_{(k)}^2$ is close to the average $\frac{1}{r} \sum_{j\in\Omega} h_{(j)}^2$:
\begin{align}
\Bigl| h_{(k)}^2 - \frac{1}{r} \sum_{j\in\Omega} h_{(j)}^2 \Bigr| < \rho / n, \qquad \forall k\in \Omega \label{eq:diff_average}
\end{align}
By \eqref{eq:sum_large} and \eqref{eq:diff_average}, for $k\in \Omega$:
\begin{align*}
h_{(k)}^2 \leq \frac{1}{r} + \frac{\rho}{n} \leq \frac{1+\rho}{r},
\end{align*}
\begin{align*}
h_{(k)}^2 \geq \frac{1-(n-r)\cdot \frac{\rho}{n}}{r} - \frac{\rho}{n} = \frac{1-\rho}{r}.
\end{align*}
Therefore,
\begin{align}
\Bigl| h_{(k)}^2 - \frac{1}{r}\Bigr| \leq \frac{\rho}{r}\qquad \forall k\in\Omega,  \label{eq:large_val}
\end{align}
It follows from \eqref{eq:small_val} and \eqref{eq:large_val} that $h$ is in the $(\rho, r)$-neighborhood of a stationary point $h_0$, where $h_{0(j)} = 1/\sqrt{r}$ if $j\in\Omega$ and $h_{0(j)} = 0$ if $j\notin\Omega$. Clearly, such an $h$ belongs to $\calH''_1\cup\calH''_2$. By contraposition, any point $h\in \calH''_3 = S^{n-1}\backslash(\calH''_1\cup\calH''_2)$ satisfies \eqref{eq:NG2}.
\end{IEEEproof}


\begin{IEEEproof}[Proof of Lemma \ref{lem:concentration}]
For any given $h\in S^{n-1}$ one can bound the deviation of the gradient (or Hessian) from its mean using matrix Bernstein inequality \cite{tropp2011user}. Let $S_\epsilon$ be an $\epsilon$-net of $S^{n-1}$. Then $|S_\epsilon|\leq (3/\epsilon)^n$ \cite[Lemma 9.5]{ledoux2013probability}. We can then bound the deviation over $S^{n-1}$ by a union bound over $S_\epsilon$.

Define $\nabla''_i \coloneqq C_{x_i}^\T \nabla_\phi(C_{x_i} h)$, and $H''_i = C_{x_i}^\T H_\phi(C_{x_i} h) C_{x_i}$. 
For the Bernoulli-Rademacher model in (A1), we have $|x_{i(j)}| \leq 1$. Therefore,
\begin{align*}
\bigl| \nabla''_{i(j)} \bigr| 
& = \Bigl| \sum_{s=1}^n x_{i(1+s-j)} \Bigl(\sum_{t=1}^n x_{i(1+s-t)}h_{(t)}\Bigr)^3 \Bigr| \\
& \leq n \Bigl(\sum_{t=1}^n |h_{(t)}|\Bigr)^3 \\
& \leq n^2\sqrt{n},
\end{align*}
\begin{align*}
\bigl| H''_{i(jk)} \bigr|
& = \Bigl| 3 \sum_{s=1}^n x_{i(1+s-j)}x_{i(1+s-k)} \Bigl(\sum_{t=1}^n x_{i(1+s-t)}h_{(t)}\Bigr)^2 \Bigr| \\
& \leq 3n \Bigl(\sum_{t=1}^n |h_{(t)}|\Bigr)^2 \\
& \leq 3n^2.
\end{align*}
It follows that $\norm{\nabla''_i} \leq n^3$, and $\norm{H''_i} \leq \norm{H''_i}_\rmF \leq 3n^3$.

Our goal is to bound the following average of independent random terms with zero mean:
\[
\nabla_{L''}(h) - \bbE\nabla_{L''}(h) = \frac{1}{N} \sum_{i=1}^N \bigl( \nabla''_i - \bbE\nabla''_i\bigr).
\]
\[
H_{L''}(h) - \bbE H_{L''}(h) = \frac{1}{N} \sum_{i=1}^N \bigl( H''_i - \bbE H''_i\bigr).
\]
Since $\norm{\nabla''_i} \leq n^3$, we have 
\begin{align*}
& \norm{\nabla''_i - \bbE\nabla''_i} \leq 2n^3, \\
&  \sum_{i=1}^N \bbE\norm{\nabla''_i - \bbE\nabla''_i}^2 = N (\bbE\norm{\nabla''_i}^2 - \norm{\bbE\nabla''_i}^2) \leq Nn^6, \\
& \Bigl\|  \sum_{i=1}^N \bbE(\nabla''_i - \bbE\nabla''_i)(\nabla''_i - \bbE\nabla''_i)^\T \Bigr\| \leq N (\bbE\norm{\nabla''_i}^2 + \norm{\bbE\nabla''_i}^2) \leq 2Nn^6. 
\end{align*}
By the rectangular version of the matrix Bernstein inequality \cite[Theorem 1.6]{tropp2011user}, and a union bound over $S_\epsilon$,
\begin{align}
& \bbP\Bigl[\sup_{h\in S_\epsilon}\norm{\nabla_{L''}(h) - \bbE\nabla_{L''}(h)} \leq \tau\Bigr]  \nonumber\\
& \geq 1 - \Bigl(\frac{3}{\epsilon}\Bigr)^n (n+1) \exp\Bigl(\frac{-N^2\tau^2/2}{2Nn^6 + 2n^3 N \tau / 3}\Bigr) \label{eq:deviation_gradient}
\end{align}
Similarly, since $\norm{H''_i}\leq 3n^3$, we have 
\begin{align*}
& \norm{H''_i - \bbE H''_i} \leq 6n^3, \\
& \Bigl\| \sum_{i=1}^N \bbE (H''_i - \bbE H''_i)^2 \Bigr\| \leq N \norm{\bbE H_i''^2 - (\bbE H''_i)^2} \leq 2N(3n^3)^2 = 18Nn^6. 
\end{align*}
By the symmetric version of the matrix Bernstein inequality \cite[Theorem 1.4]{tropp2011user}, and a union bound over $S_\epsilon$,
\begin{align}
& \bbP\Bigl[\sup_{h\in S_\epsilon}\norm{H_{L''}(h) - \bbE H_{L''}(h)} \leq \tau\Bigr]  \nonumber\\
& \geq 1 - \Bigl(\frac{3}{\epsilon}\Bigr)^n (2n) \exp\Bigl(\frac{-N^2\tau^2/2}{18Nn^6 + 6n^3 N \tau / 3}\Bigr) \label{eq:deviation_Hessian}
\end{align}
Choose $\tau = \frac{\theta(1-3\theta)\rho^2}{8n}$, and $\epsilon = \frac{\tau}{6n^3} = \frac{\theta(1-3\theta)\rho^2}{48n^4}$. By \eqref{eq:deviation_gradient} and \eqref{eq:deviation_Hessian}, there exist constants $c_2,\, c_2' > 0$ (depending only on $\theta$), such that: if $N > \frac{c_2 n^9}{\rho^4}\log \frac{n}{\rho}$, then with probability at least $1 - e ^{-c_2'n}$,
\[
\sup_{h\in S_\epsilon}\norm{\nabla_{L''}(h) - \bbE\nabla_{L''}(h)} \leq \tau = \frac{\theta(1-3\theta)\rho^2}{8n},
\]
\[
\sup_{h\in S_\epsilon}\norm{H_{L''}(h) - \bbE H_{L''}(h)} \leq \tau = \frac{\theta(1-3\theta)\rho^2}{8n}.
\]

To finish the proof, we extrapolate the concentration bounds over $S_\epsilon$ to all points in $S^{n-1}$. For any $h \in S^{n-1}$, there exists $h' \in S_\epsilon$ such that $\norm{h-h'} \leq \epsilon$. Furthermore, thanks to the Lipschitz continuity of the gradient and the Hessian,
\begin{align*}
& \norm{\nabla''_i(h) - \nabla''_i(h')} \\
& \leq \norm{C_{x_i}} \cdot \sqrt{n} (3\norm{x_i}^2) \cdot \norm{x_i} \norm{h-h'} \\
& \leq 3n^3 \epsilon, 
\end{align*} 
\begin{align*}
& \norm{H''_i(h) - H''_i(h')} \\
& \leq \norm{C_{x_i}}^2 \cdot (6\norm{x_i}) \cdot \norm{x_i} \norm{h-h'} \\
& \leq 6n^3 \epsilon, 
\end{align*}
where $3\norm{x_i}^2$ and $6\norm{x_i}$ are the Lipschitz constants of $(\cdot)^3$ and $3(\cdot)^2$ on the interval $[-\norm{x_i},\norm{x_i}]$. We also use the fact that $|x_{i(j)}| < 1$, hence $\norm{x_i} \leq \sqrt{n}$ and $\norm{C_{x_i}}\leq n$.
As a consequence,
\begin{align*}
& \sup_{h\in S^{n-1}}\norm{\nabla_{L''}(h) - \bbE\nabla_{L''}(h)} \\
& \leq \sup_{h\in S_\epsilon}\norm{\nabla_{L''}(h) - \bbE\nabla_{L''}(h)} + 2 \max_{i\in[N]}\sup_{\norm{h-h'}\leq \epsilon} \norm{\nabla''_i(h) - \nabla''_i(h')} \\
& \leq \tau + 6n^3 \epsilon = 2\tau = \frac{\theta(1-3\theta)\rho^2}{4n},
\end{align*}
\begin{align*}
& \sup_{h\in S^{n-1}} \norm{\hn_{L''}(h) - \bbE \hn_{L''}(h)} \\
& \leq \sup_{h\in S^{n-1}} \norm{\nabla_{L''}(h) - \bbE \nabla_{L''}(h)} \\
& \leq \frac{\theta(1-3\theta)\rho^2}{4n}.
\end{align*}
Similarly,
\begin{align*}
& \sup_{h\in S^{n-1}}\norm{H_{L''}(h) - \bbE H_{L''}(h)} \\
& \leq \sup_{h\in S_\epsilon}\norm{H_{L''}(h) - \bbE H_{L''}(h)} + 2 \max_{i\in[N]} \sup_{\norm{h-h'}\leq \epsilon} \norm{H''_i(h) - H''_i(h')} \\
& \leq \tau + 12n^3 \epsilon = 3\tau = \frac{3\theta(1-3\theta)\rho^2}{8n},
\end{align*}
\begin{align*}
& \sup_{h\in S^{n-1}}\norm{\hH_{L''}(h) - \bbE \hH_{L''}(h)} \\
& \leq \sup_{h\in S^{n-1}}\norm{H_{L''}(h) - \bbE H_{L''}(h)} + \sup_{h\in S^{n-1}}\norm{\nabla_{L''}(h) - \bbE \nabla_{L''}(h)} \\
& \leq \frac{\theta(1-3\theta)\rho^2}{n}.
\end{align*}
\end{IEEEproof}


\begin{IEEEproof}[Proof of Lemma \ref{lem:not_flat}]
We have $\bbE \frac{1}{\theta n N} \sum_{i=1}^N C_{x_i}^\T C_{x_i} = I$. 
We first bound $\norm{\frac{1}{\theta n N} \sum_{i=1}^N C_{x_i}^\T C_{x_i} - I}$ using the matrix Bernstein inequality. To this end, we bound the spectral norm of $\bbE (C_{x_i}^\T C_{x_i})^2$, the eigenvalues of which can be computed using the DFT of $x_i$. The eigenvalue corresponding to the $t$-th frequency satisfies
\begin{align*}
& \bbE \Bigl[\Bigl(\sum_{k=1}^n e^{-\sqrt{-1}(k-1)t/n}x_{i(k)}\Bigr) \Bigl(\sum_{k=1}^n e^{\sqrt{-1}(k-1)t/n}x_{i(k)}\Bigr) \Bigr]^2 \\
& = \bbE \Bigl(\sum_{k=1}^n x_{i(k)}^2 + \sum_{1\leq k < j\leq n} 2\cos((j-k)t/n) x_{i(j)}x_{i(k)} \Bigr)^2 \\
&  \leq n \theta + \frac{n(n-1)}{2} \times 4\theta^2 + n(n-1) \theta^2 \\
& = n\theta + 3n(n-1) \theta^2.
\end{align*}
Therefore,
\begin{align*}
& \Bigl\| \sum_{i=1}^N \bbE \Bigl(\frac{1}{\theta n} C_{x_i}^\T C_{x_i} - I \Bigr)^2 \Bigr\| \\& = N \Bigl\| \frac{1}{\theta^2 n^2} \bbE (C_{x_i}^\T C_{x_i})^2 - I \Bigr\| \\
& \leq \frac{N}{\theta^2 n^2} \norm{\bbE (C_{x_i}^\T C_{x_i})^2} + N \\
& \leq \frac{N}{\theta^2 n^2} (n\theta + 3n(n-1) \theta^2) + N \\
& \leq \frac{N}{\theta n} + 3N + N \\
& \leq 5N.
\end{align*}
We also have
\[
\norm{\frac{1}{\theta n} C_{x_i}^\T C_{x_i} - I} \leq \frac{1}{\theta n} \norm{C_{x_i}}^2 + 1 \leq \frac{n^2}{\theta n} + 1 \leq n^2 + 1.
\]
By the matrix Bernstein inequality \cite[Theorem 1.4]{tropp2011user},
\[
\bbP\Bigl[ \Bigl\| \frac{1}{\theta n N} \sum_{i=1}^N C_{x_i}^\T C_{x_i} - I \Bigr\| \leq \tau \Bigr] \geq 1 - 2n \exp\Bigl( \frac{-N^2 \tau^2 / 2}{5N + (n^2+1) N\tau / 3} \Bigr).
\]
Set $\tau = \frac{\theta(1-3\theta)\rho^2}{200n^4 \kappa^4}$. Then there exist constants $c_3,\, c_3' > 0$ (depending only on $\theta$) such that: if $N > \frac{c_3 n^8 \kappa^8}{\rho^4} \log n$, then with probability at least $1-n^{-c_3'}$,
\begin{align}
\Bigl\| \frac{1}{\theta n N} \sum_{i=1}^N C_{x_i}^\T C_{x_i} - I \Bigr\| \leq \frac{\theta(1-3\theta)\rho^2}{200n^4 \kappa^4}. \label{eq:near_identity}
\end{align}

Next, we bound $\norm{C_f R - C_f(C_f^\T C_f)^{-1/2}}$ (similar to the proofs of \cite[Lemma 15]{sun2017complete} and \cite[Lemma B.2]{sun2015complete}). Define $Q\coloneqq \frac{1}{\theta n N}\sum_{i=1}^N C_{x_i}^\T C_{x_i}$. Then
\begin{align}
& \norm{C_f R - C_f(C_f^\T C_f)^{-1/2}}  \nonumber\\
& = \norm{C_f (C_f^\T Q C_f)^{-1/2} - C_f(C_f^\T C_f)^{-1/2}}  \nonumber\\
& \leq \sigma_1(C_f) \cdot \norm{(C_f^\T Q C_f)^{-1/2} - (C_f^\T C_f)^{-1/2}}  \nonumber\\
& \leq \sigma_1(C_f) \frac{\norm{ (C_f^\T Q C_f)^{-1} - (C_f^\T C_f)^{-1} } }{\sigma_n\bigl((C_f^\T C_f)^{-1/2}\bigr)}  \label{eq:sqrt_bound}\\
& = \sigma_1^2(C_f) \norm{ (C_f^\T Q C_f)^{-1} - (C_f^\T C_f)^{-1} }  \nonumber\\
& \leq \frac{\sigma_1^2(C_f) }{\sigma_n^2(C_f) } \norm{ (C_f^\T C_f)(C_f^\T Q C_f)^{-1} - I }  \nonumber\\
& = \kappa^2  \Bigl\| \Bigl[I + (C_f^\T(Q-I)C_f)(C_f^\T C_f)^{-1} \Bigr]^{-1} - I \Bigr\|  \nonumber\\
& \leq \kappa^2 \frac{\norm{C_f^\T(Q-I)C_f}\norm{(C_f^\T C_f)^{-1}}}{1- \norm{C_f^\T(Q-I)C_f}\norm{(C_f^\T C_f)^{-1}}}  \label{eq:inverse_bound}\\
& \leq \kappa^4 \frac{\norm{Q-I}}{1 - 1/2} \leq \frac{\theta(1-3\theta)\rho^2}{100n^4}. \label{eq:near_orthogonal}
\end{align}
The inequality \eqref{eq:sqrt_bound} follows from the fact (\cite[Theorem 6.2]{higham2008functions}) that, for positive definite $A$ and $B$,
\[
\norm{A^{-1/2}-B^{-1/2}} \leq \frac{\norm{A^{-1} - B^{-1}}}{\sigma_n(A^{-1/2}+B^{-1/2})} \leq \frac{\norm{A^{-1} - B^{-1}}}{\sigma_n(B^{-1/2})},
\]
which in turn follows from the identity 
\[
(A^{-1/2}-B^{-1/2})(A^{-1/2}+B^{-1/2}) + (A^{-1/2}+B^{-1/2})(A^{-1/2}-B^{-1/2}) = 2(A^{-1} - B^{-1}).
\]
The inequality \eqref{eq:inverse_bound} is due to the fact that $\norm{(I+A)^{-1} - I}\leq \norm{(I+A)^{-1}} \norm{A} \leq \frac{\norm{A}}{1 - \norm{A}}$ for $\norm{A} < 1$. The last line \eqref{eq:near_orthogonal} follows from \eqref{eq:near_identity} and $\norm{C_f^\T(Q-I)C_f}\norm{(C_f^\T C_f)^{-1}} \leq \kappa^2 \norm{Q-I} < \frac{1}{2}$.

The rest of Lemma \ref{lem:not_flat} follows from the Lipschitz continuity of the objective function. Define $U\coloneqq C_fR$, and $U'\coloneqq C_f(C_f^\T C_f)^{-1/2}$, which is an orthogonal matrix. We have
\begin{align}
\norm{C_fR} = \norm{U} \leq \norm{U'} + \norm{U-U'} < 2. \label{eq:CfR_norm}
\end{align}
Recall that for the Bernoulli-Rademacher model, $\norm{x_i}\leq \sqrt{n}$ and $\norm{C_{x_i}}\leq n$. Then the difference of the gradients of $L(h) = \frac{1}{N}\sum_{i=1}^N \phi(C_{x_i}Uh)$ and $L'(h) =\frac{1}{N}\sum_{i=1}^N \phi(C_{x_i}U'h)$ can be bounded as follows:
\begin{align*}
& \norm{\nabla_L(h) - \nabla_{L'}(h)} \\
& \leq \max_{i\in [N]} \norm{ U^\T C_{x_i}^\T \nabla_\phi(C_{x_i}Uh) - U'^\T C_{x_i}^\T \nabla_\phi(C_{x_i}U'h) } \\
& \leq \max_{i\in [N]} \norm{ U^\T C_{x_i}^\T \nabla_\phi(C_{x_i}Uh) - U^\T C_{x_i}^\T \nabla_\phi(C_{x_i}U'h) } \\
& \quad + \max_{i\in [N]} \norm{ U^\T C_{x_i}^\T \nabla_\phi(C_{x_i}U'h) - U'^\T C_{x_i}^\T \nabla_\phi(C_{x_i}U'h) }  \\
& \leq \max_{i\in [N]} \norm{U} \norm{C_{x_i}} \cdot \sqrt{n} [3 (\norm{U}\norm{x_i})^2] \cdot \norm{U-U'} \norm{x_i} \\
& \quad + \max_{i\in [N]} \norm{U-U'} \norm{C_{x_i}} \cdot \sqrt{n} \norm{x_i}^3 \\
& \leq 25\sqrt{n} \cdot \max_{i\in [N]} \norm{C_{x_i}}\norm{x_i}^3  \norm{U-U'} \\
& \leq 25 n^3 \norm{U-U'},
\end{align*}
where the third inequality follows from the fact that $\nabla_\phi(\cdot)$ is Lipschitz continous and bounded on compact sets -- the Lipschitz constant of $(\cdot)^3$ on the interval $[-\norm{U}\norm{x_i})^2,\norm{U}\norm{x_i})^2]$ is $3 (\norm{U}\norm{x_i})^2$, and the upper bound of $|(\cdot)^3|$ on the interval $[-\norm{x_i},\norm{x_i}]$ is $\norm{x_i}^3$.
Similarly the difference of the Hessians can be bounded as follows:
\begin{align*}
& \norm{H_L(h) - H_{L'}(h)} \\
& \leq \max_{i\in [N]} \norm{ U^\T C_{x_i}^\T H_\phi(C_{x_i}Uh)C_{x_i}U - U'^\T C_{x_i}^\T H_\phi(C_{x_i}U'h)C_{x_i}U' } \\
& \leq \max_{i\in [N]} \norm{ U^\T C_{x_i}^\T H_\phi(C_{x_i}Uh)C_{x_i}U - U^\T C_{x_i}^\T H_\phi(C_{x_i}U'h)C_{x_i}U } \\
& \quad + \max_{i\in [N]} \norm{ U^\T C_{x_i}^\T H_\phi(C_{x_i}U'h)C_{x_i}U - U'^\T C_{x_i}^\T H_\phi(C_{x_i}U'h)C_{x_i}U }  \\
& \quad + \max_{i\in [N]} \norm{ U'^\T C_{x_i}^\T H_\phi(C_{x_i}U'h)C_{x_i}U - U'^\T C_{x_i}^\T H_\phi(C_{x_i}U'h)C_{x_i}U' }  \\
& \leq \max_{i\in [N]} \norm{U}^2 \norm{C_{x_i}}^2 \cdot [6 (\norm{U}\norm{x_i})] \cdot \norm{U-U'} \norm{x_i} \\
& \quad + \max_{i\in [N]} \norm{U-U'} \norm{U} \norm{C_{x_i}}^2 \cdot [3 \norm{x_i}^2] \\
& \quad + \max_{i\in [N]} \norm{U-U'} \norm{C_{x_i}}^2 \cdot [3 \norm{x_i}^2] \\
& \leq 57 \cdot \max_{i\in [N]} \norm{C_{x_i}}^2\norm{x_i}^2  \norm{U-U'} \\
& \leq 57 n^3 \norm{U-U'},
\end{align*}
where the third inequality uses the Lipschitz constant and upper bound of $3(\cdot)^2$.

It follows from \eqref{eq:near_orthogonal} and the above bounds that
\begin{align*}
& \sup_{h\in S^{n-1}} \norm{\hn_L(h) - \hn_{L'}(h)} \\
& \leq \sup_{h\in S^{n-1}} \norm{\nabla_L(h) - \nabla_{L'}(h)} \\
& \leq 25 n^3 \norm{U-U'} \leq \frac{\theta(1-3\theta)\rho^2}{4n}.
\end{align*}
\begin{align*}
& \sup_{h\in S^{n-1}} \norm{\hH_L(h) - \hH_{L'}(h)} \\
& \leq \sup_{h\in S^{n-1}} \norm{H_L(h) - H_{L'}(h)} + \sup_{h\in S^{n-1}} \norm{\nabla_L(h) - \nabla_{L'}(h)} \\
& \leq 100 n^3 \norm{U-U'} \leq \frac{\theta(1-3\theta)\rho^2}{n}.
\end{align*}

\end{IEEEproof}


\begin{IEEEproof}[Proof of Lemma \ref{lem:neighborhood_local_minima}]
The set $\calH''_1$ equals the union of $(\rho, 1)$-neighborhoods of $\{\pm e_j\}_{j=1}^n$, and the columns of $C_f^{-1}=C_g$ are the shifted versions of the inverse filter $g$. Therefore, by \eqref{eq:neighbor_diff}, every point $h^* \in (C_f^\T C_f)^{1/2} C_f^{-1} \calH''_1$ satisfies
\[
\norm{C_f(C_f^\T C_f)^{-1/2} h^* \pm e_j} \leq \sqrt{2\rho},
\]
for some $j\in[n]$. 
It follows that
\begin{align}
& \norm{C_fR h^* \pm e_j}   \nonumber\\
& \leq \norm{C_fR h^* - C_f(C_f^\T C_f)^{-1/2} h^*} + \norm{C_f(C_f^\T C_f)^{-1/2} h^* \pm e_j}  \nonumber\\
& \leq \frac{\theta(1-3\theta)\rho^2}{100n^4} + \sqrt{2\rho} \nonumber\\
& \leq 2\sqrt{\rho},  \nonumber
\end{align}
where the second to last line follows from \eqref{eq:near_orthogonal}, and the last line follows from $\theta(1-3\theta)\rho^2 / (100n^4) < (2-\sqrt{2})\sqrt{\rho}$.
\end{IEEEproof}


\section{Proofs for Section \ref{sec:optimization}}

\begin{IEEEproof}[Proof of Lemma \ref{lem:L_bound}]
Clearly, $L(h) \leq 0$ for all $h\in S^{n-1}$. 
For the Bernoulli-Rademacher model in (A1), we have $\norm{x_i} \leq \sqrt{n}$ and $\norm{C_{x_i}} \leq n$. Therefore,
\begin{align*}
L(h) & \geq -\frac{1}{4} \max_{i\in[N]} \norm{C_{x_i}C_f R h}_4^4\\
& \geq -\frac{n}{4} \max_{i\in[N]}(\norm{x_i}\norm{U h})^4 \\
& \geq - 4n^3,
\end{align*}
where the second inequality follows from the Cauchy-Schwarz inequality, and the third inequality follows from $\norm{U} = \norm{C_fR}\leq 2$ (see \eqref{eq:CfR_norm}). 

We can bound the the norm of $\nabla_L(h)$ and $H_L(h)$ similarly. To bound $\norm{\nabla_L(h)}$, we have
\begin{align*}
\norm{\nabla_L(h)}  
& \leq \max_{i\in[N]} \norm{U^\T C_{x_i}^\T \nabla_\phi(C_{x_i} U h)} \\ 
& \leq \max_{i\in[N]} \norm{U}\norm{C_{x_i}} \cdot \sqrt{n} (\norm{x_i}\norm{U})^3\\
& \leq 16 n^3.
\end{align*}
To bound $\norm{H_L(h)}$, we have
\begin{align*}
\norm{H_L(h)}  
& \leq \max_{i\in[N]} \norm{U^\T C_{x_i}^\T H_\phi(C_{x_i}U h) C_{x_i} U } \\ 
& \leq \max_{i\in[N]} \norm{U}^2\norm{C_{x_i}}^2 \times 3(\norm{x_i} \norm{U})^2 \\
& \leq 48 n^3.
\end{align*}

\end{IEEEproof}


\begin{IEEEproof}[Proof of Lemma \ref{lem:lipschitz}]
We first derive Lipschitz constants for the gradient and the Hessian. Recall that $\norm{U} = \norm{C_f R} \leq 2$ (see \eqref{eq:CfR_norm}). Therefore, the difference of two gradients can be bounded as follows:
\begin{align*}
& \norm{\nabla_L(h) - \nabla_{L}(h')} \\
& \leq \max_{i\in [N]} \norm{ U^\T C_{x_i}^\T \nabla_\phi(C_{x_i}Uh) - U^\T C_{x_i}^\T \nabla_\phi(C_{x_i}Uh') } \\
& \leq \max_{i\in [N]} \norm{U} \norm{C_{x_i}} \cdot \sqrt{n} [3 (\norm{U}\norm{x_i})^2] \cdot \norm{x_i} \norm{U} \norm{h - h'}\\
& \leq 48 n^3 \norm{h-h'},
\end{align*}
where the second inequality follows from the fact that $\nabla_\phi(\cdot)$ is Lipschitz continuous -- the Lipschitz constant of $(\cdot)^3$ on the interval $[-\norm{U}\norm{x_i})^2,\norm{U}\norm{x_i})^2]$ is $3 (\norm{U}\norm{x_i})^2$.
Similarly the difference of two Hessians can be bounded as follows:
\begin{align*}
& \norm{H_L(h) - H_{L}(h')} \\
& \leq \max_{i\in [N]} \norm{ U^\T C_{x_i}^\T H_\phi(C_{x_i}Uh)C_{x_i}U - U^\T C_{x_i}^\T H_\phi(C_{x_i}Uh')C_{x_i}U } \\
& \leq \max_{i\in [N]} \norm{U}^2 \norm{C_{x_i}}^2 \cdot [6 (\norm{U}\norm{x_i})] \cdot \norm{x_i} \norm{U} \norm{h-h'} \\
& \leq 96 n^3 \norm{h-h'},
\end{align*}
where the second inequality uses the Lipschitz constant and upper bound of $3(\cdot)^2$.

It follows from the definitions of Riemannian gradient and Riemannian Hessian that they are Lipschitz continuous:
\begin{align*}
& \norm{\hn_L(h) - \hn_{L}(h')} \\
& \leq \norm{P_{h^\perp}(\nabla_L(h) - \nabla_{L}(h'))} + \norm{(P_{h^\perp} - P_{h'^\perp})\nabla_{L}(h')} \\
& \leq 48 n^3 \norm{h-h'} + 16 n^3  \norm{P_{h^\perp} - P_{h'^\perp}} \\
& \leq 64 n^3 \norm{h-h'}.
\end{align*}
\begin{align*}
& \norm{\hH_L(h) - \hH_{L}(h')} \\
& \leq \norm{P_{h^\perp}(H_L(h) - H_{L}(h'))P_{h^\perp}} + \norm{P_{h^\perp}H_{L}(h') (P_{h^\perp} - P_{h'^\perp})} + \norm{(P_{h^\perp} - P_{h'^\perp})H_{L}(h') P_{h'^\perp}} \\
& \quad + \norm{(\hn_L(h) - \hn_{L}(h'))h^\T} + \norm{\hn_{L}(h')(h-h')^\T}\\
& \leq 96 n^3 \norm{h-h'} + 48n^3 \norm{h-h'} \times 2 + 64n^3\norm{h-h'} + 16n^3 \norm{h-h'}\\
& = 272 n^3 \norm{h-h'}.
\end{align*}

\end{IEEEproof}


\begin{IEEEproof}[Proof of Lemma \ref{lem:invertible_differential}]
Suppose the columns of matrix $V\in\bbR^{n\times(n-1)}$ (resp. $V'\in\bbR^{n\times(n-1)}$) form a orthonormal basis for the tangent subspace at $h$ (resp. $h'$). Then a matrix representation of $D\calA(h)$ in \eqref{eq:differential} as a mapping from the tangent space of $h$ to the tangent space at $h'$ with respect to the bases of these spaces is $V'^\T V (I_{n-1} - \gamma V^\T \hH_L(h) V)$. 

Note that $|\mathrm{det}(V'^\T V)|$ does not depend on the specific choice of orthogonal bases $V$ and $V'$ (multiplication by an orthonormal matrix does not change $|\mathrm{det}(\cdot)|$). Therefore, we consider the following construction of $V$ and $V'$. Suppose the columns of $V_\cap\in\bbR^{n\times (n-2)}$ form an orthonormal basis for the intersection of the tangent spaces at $h$ and at $h'$. Let $c\coloneqq \langle h, h'\rangle < 1$, then it is easy to verify that $V \coloneqq [V_\cap,\, \frac{1}{\sqrt{1-c^2}}(h' - c h)]$ and $V' \coloneqq [V_\cap,\, \frac{1}{\sqrt{1-c^2}}(ch' - h)]$ are valid orthonormal bases. It follows that
\[
|\mathrm{det}(V'^\T V)| = \begin{vmatrix}
I_{n-2} & \bfz_{(n-2)\times 1}\\
\bfz_{1\times(n-2)} & c 
\end{vmatrix} = |c|
\]
Since $\langle h, h' \rangle = \langle h,\, h-\gamma \hn_L(h) \rangle /\norm{h-\gamma \hn_L(h)} = \norm{h}^2 /\norm{h-\gamma \hn_L(h)} = 1 /\norm{h-\gamma \hn_L(h)}> 0$, we have $|\mathrm{det}(V'^\T V)| = |\langle h, h' \rangle| > 0$. 

By Lemma \ref{lem:L_bound}, for all $h \in S^{n-1}$, 
\begin{align*}
& \norm{\hH_L(h)} \leq \norm{H_L(h)} + \norm{\nabla_L(h)} \\
& \leq 48 n^3 + 16 n^3 = 64 n^3.
\end{align*}
Therefore $I_{n-1} - \gamma V^\T \hH_L(h) V$ is strictly positive definite for $\gamma < 1/(64n^3)$. 

It follows that
\[
|\mathrm{det}(D\calA(h))| = |\mathrm{det}(V'^\T V)| \cdot |\mathrm{det}\bigl(I_{n-1} - \gamma V^\T H_L(h) V \bigr)| > 0.
\]
\end{IEEEproof}


\begin{IEEEproof}[Proof of Lemma \ref{lem:gradient_descent_step}]
One can write the left-hand side as a line integral along the shortest path between $h^{(t)}$ and $h^{(t+1)}$ on the sphere,
\begin{align*}
& L(h^{(t+1)}) - L(h^{(t)}) \\
& = \int_{h^{(t)}}^{h^{(t+1)}} ~d L(h) \\
& = \int_{h^{(t)}}^{h^{(t+1)}} \hn_L(h) ~d h \\
& = \langle\hn_L(h^{(t)}), h^{(t+1)} - h^{(t)}\rangle + \int_{h^{(t)}}^{h^{(t+1)}} (\hn_L(h)-\hn_L(h^{(t)})) ~d h \\
& \leq \langle\hn_L(h^{(t)}), h^{(t+1)} - h^{(t)}\rangle + \norm{\hn_L(h^{(t+1)})-\hn_L(h^{(t)})} \int_{h^{(t)}}^{h^{(t+1)}} ~\norm{d h} \\
& \leq \langle\hn_L(h^{(t)}), h^{(t+1)} - h^{(t)}\rangle + 64n^3\norm{h^{(t+1)}- h^{(t)} } \int_{h^{(t)}}^{h^{(t+1)}} ~\norm{d h},
\end{align*}
where the first inequality uses Cauchy-Schwarz inequality, and the second inequality uses Lipschitz continuity of the Riemannian gradient.
Let $\Theta \coloneqq \angle (h^{(t)},h^{(t+1)}) = \tan^{-1} (\gamma\norm{  \hn_L(h^{(t)})})$, then 
\begin{align*}
& L(h^{(t+1)}) - L(h^{(t)}) \\
& \leq -\gamma \norm{\hn_L(h^{(t)})}^2 \cos\Theta + 64n^3 \gamma^2 \norm{\hn_L(h^{(t)})}^2 \frac{2\Theta \sin\frac{\Theta}{2} }{\tan^2 \Theta} \\
& \leq (-\gamma\cos\Theta + 64n^3 \gamma^2) \norm{\hn_L(h^{(t)})}^2\\
& \leq -\frac{0.0038}{n^3} \norm{\hn_L(h^{(t)})}^2,
\end{align*}
where the second inequality follows from $2\sin\frac{\Theta}{2} \leq \Theta \leq \tan\Theta$, and the last step follows from $\cos\Theta \geq {1}/{\sqrt{1+(\frac{1}{128n^3}\times 16n^3)^2}} \approx 0.99$.
\end{IEEEproof}


\begin{proof}[Proof of Lemma \ref{lem:improve_or_localize}]
\begin{align*}
& \norm{h^{(0)} - h^{(t)}} \\
& \leq \sum_{\tau=1}^{t}  \norm{h^{(\tau-1)} - h^{(\tau)}} \\
& \leq \sum_{\tau=1}^{\calT}  \norm{h^{(\tau-1)} - h^{(\tau)}} \\
& \leq \sqrt{\calT \sum_{\tau=1}^{\calT} \norm{h^{(\tau-1)} - h^{(\tau)}}^2} \\
& \leq \sqrt{\calT \gamma^2 \sum_{\tau=1}^{\calT} \norm{\hn_L(h^{(\tau-1)})}^2} \\
& \leq \sqrt{\calT (\frac{1}{128 n^3})^2 \frac{n^3}{0.0038} \sum_{\tau=1}^{\calT} L(h^{(\tau-1)})-L(h^{(\tau)})}\\
& \leq \sqrt{\frac{\calT}{62n^3}[L(h^{(0)}) - L(h^{(\calT)})]},
\end{align*}
where the first inequality follows from triangle inequality, the third inequality follows from the Cauchy-Schwarz inequality, and the fifth inequality follows from Lemma \ref{lem:gradient_descent_step}.

\end{proof}


\begin{IEEEproof}[Proof of Lemma \ref{lem:coupling_sequence}]
Assume that the contrary is true, i.e., $\max\{ L(h^{(0)}_1) - L(h^{(\calT)}_1),\,L(h^{(0)}_2) - L(h^{(\calT)}_2) \} \leq \calL$. By Lemma \ref{lem:improve_or_localize}, $\max\{ \norm{h^{(0)}_1 - h^{(t)}_1},\,\norm{h^{(0)}_2 - h^{(t)}_2} \} \leq \sqrt{\frac{\calT\calL}{62n^3}}$. Therefore for all $0\leq t \leq \calT$:
\[
\max\{ \norm{h_0 - h^{(t)}_1},\,\norm{h_0 - h^{(t)}_2} \} \leq \sqrt{\frac{\calT\calL}{62n^3}} + \frac{\calR}{2} \leq \calR = \frac{\theta(1-3\theta)}{\xi n^3\log n},
\]
and
\begin{align}
\norm{h^{(\calT)}_1 -h^{(\calT)}_2} \leq 2\calR.
\label{eq:localize_coupled_sequence}
\end{align}

On the other hand, manifold gradient descent satisfies
\begin{align*}
& h^{(t)}_1 -h^{(t)}_2 \\
& = \frac{(h^{(t-1)}_1 - \gamma\hn_L(h^{(t-1)}_1)) - (h^{(t-1)}_2 - \gamma\hn_L(h^{(t-1)}_2))}{\norm{h^{(t-1)}_1 - \gamma \hn_L(h^{(t-1)}_1)}} \\
& \quad + (h^{(t-1)}_2 - \gamma\hn_L(h^{(t-1)}_2)) \Bigl( \frac{1}{\norm{h^{(t-1)}_1 - \gamma \hn_L(h^{(t-1)}_1)}} - \frac{1}{\norm{h^{(t-1)}_2 - \gamma \hn_L(h^{(t-1)}_2)}} \Bigr).
\end{align*}
The numerator of the first term is 
\begin{align*}
& (h^{(t-1)}_1 - \gamma\hn_L(h^{(t-1)}_1)) - (h^{(t-1)}_2 - \gamma\hn_L(h^{(t-1)}_2)) \\
& = (h^{(t-1)}_1 - h^{(t-1)}_2) + \gamma (\hn_L(h^{(t-1)}_2) - \hn_L(h^{(t-1)}_1)) \\
& = (h^{(t-1)}_1 - h^{(t-1)}_2) + \gamma \int_{h^{(t-1)}_1}^{h^{(t-1)}_2} \hH_L(h) dh  \\
& = (I-\gamma \hH_L(h_0))(h^{(t-1)}_1 - h^{(t-1)}_2) + \gamma \int_{h^{(t-1)}_1}^{h^{(t-1)}_2} [\hH_L(h) - \hH_L(h_0)] dh.
\end{align*}
Define 
\[
\calP^{(t)} \coloneqq \frac{I-\gamma \hH_L(h_0)}{\norm{h^{(t)}_1 - \gamma \hn_L(h^{(t)}_1)}},
\] 
and 
\begin{align*}
& q^{(t)}  \coloneqq \frac{\gamma P_{h_0^\perp}\int_{h^{(t)}_1}^{h^{(t)}_2} [\hH_L(h) - \hH_L(h_0)] dh}{\norm{h^{(t)}_1 - \gamma \hn_L(h^{(t)}_1)}} \\
& \quad + P_{h_0^\perp}(h^{(t)}_2 - \gamma\hn_L(h^{(t)}_2)) \Bigl( \frac{1}{\norm{h^{(t)}_1 - \gamma \hn_L(h^{(t)}_1)}} - \frac{1}{\norm{h^{(t)}_2 - \gamma \hn_L(h^{(t)}_2)}} \Bigr).
\end{align*}
Then
\begin{align*}
& P_{h_0^\perp} (h^{(t)}_1 -h^{(t)}_2) \\
& = \calP^{(t-1)} P_{h_0^\perp} (h^{(t-1)}_1 -h^{(t-1)}_2) + q^{(t-1)} \\
& = \Bigl[\prod_{\tau = 0}^{t-1}\calP^{(\tau)}\Bigr] P_{h_0^\perp} (h^{(0)}_1 -h^{(0)}_2) 
+ \sum_{\tau = 0}^{t-1} \Bigl[\prod_{\tau' = \tau+1}^{t-1} \calP^{(\tau')}\Bigr] q^{(\tau)}.
\end{align*}

We also have the following claims, which we prove later in this section.
\begin{claim} \label{cla:eigen_of_Hessian}
$\Bigl\|\Bigl[\prod_{\tau = 0}^{t-1}\calP^{(\tau)}\Bigr] P_{h_0^\perp} (h^{(0)}_1 -h^{(0)}_2)\Bigr\| = \Bigl[\prod_{\tau = 0}^{t-1} \norm{\calP^{(\tau)}}\Bigr] \norm{P_{h_0^\perp} (h^{(0)}_1 -h^{(0)}_2)} = d [1 + \theta(1-3\theta)\gamma]^t$.
\end{claim}
\begin{claim} \label{cla:less_than_half}
$\Bigl\|\sum_{\tau = 0}^{t-1} \Bigl[\prod_{\tau' = \tau+1}^{t-1} \calP^{(\tau')}\Bigr] q^{(\tau)}\Bigr\| 
\leq \frac{1}{2}\Bigl\|\Bigl[\prod_{\tau = 0}^{t-1}\calP^{(\tau)}\Bigr] P_{h_0^\perp} (h^{(0)}_1 -h^{(0)}_2)\Bigr\|$ for $t= 0,1,\dots, \calT$.
\end{claim}

Given the above claims, we have
\begin{align*}
& \norm{ P_{h_0^\perp} (h^{(\calT)}_1 -h^{(\calT)}_2) } \\
& \geq \Bigl\|\Bigl[\prod_{\tau = 0}^{\calT-1}\calP^{(\tau)}\Bigr] P_{h_0^\perp} (h^{(0)}_1 -h^{(0)}_2)\Bigr\| - \Bigl\|\sum_{\tau = 0}^{\calT-1} \Bigl[\prod_{\tau' = \tau+1}^{\calT-1} \calP^{(\tau')}\Bigr] q^{(\tau)}\Bigr\| \\
& \geq \frac{1}{2}\Bigl\|\Bigl[\prod_{\tau = 0}^{\calT-1}\calP^{(\tau)}\Bigr] P_{h_0^\perp} (h^{(0)}_1 -h^{(0)}_2)\Bigr\| \\
& \geq \frac{d}{2} \bigl[ 1 + \theta(1-3\theta)\gamma \bigr]^\calT \\
& \geq \frac{d}{2} (2.5)^{\calT\theta(1-3\theta)\gamma}  > \frac{d}{2} n^{\xi / 10000} = 2\calR,
\end{align*}
This clearly contradicts the localization result \eqref{eq:localize_coupled_sequence}, which finishes the proof.

\end{IEEEproof}


\begin{IEEEproof}[Proof of Claim \ref{cla:eigen_of_Hessian}]
First, we bound the denominator of $\calP^{(t)}$:
\begin{align*}
& \norm{h^{(t)}_1 - \gamma \hn_L(h^{(t)}_1)} \\
& \leq \sqrt{1 + \gamma^2\norm{\hn_L(h^{(t)}_1)}^2} \\
& \leq 1 + \frac{\gamma}{2} \norm{\hn_L(h^{(t)}_1)} \\
& \leq 1 + \frac{\gamma}{2} (\norm{\hn_L(h_0)} + \norm{\hn_L(h_0) - \hn_L(h^{(t)}_1)}) \\
& \leq 1 + \frac{\gamma}{2} (c(n,\theta,\rho) + 64n^3\calR)\\
& \leq 1 + \frac{\calR}{2},
\end{align*}
where the second inequality follows from the fact that $\sqrt{1+x^2} < 1+ x/2$ for all $0 < x < 1$, and the last inequality is due to $c(n,\theta,\rho) = \frac{\theta(1-3\theta)\rho^2}{2n} < \frac{64\theta(1-3\theta)}{\xi n} < 64n^3\calR$.

Since $P_{h_0^\perp} (h^{(0)}_1 -h^{(0)}_2) = d z_0$ is the eigenvector corresponding to the smallest eigenvalue of $\hH_L(h_0)$, it is also the dominant eigenvector of $\calP^{(t)}$. By \eqref{eq:negative_curvature}, $z_0^\T \hH_{L}(h) z_0 \leq -1.2\theta(1-3\theta)$. Hence the dominant eigenvalue is
\begin{align*}
& \norm{\calP^{(t)}} \geq 
\frac{1+1.2\theta(1-3\theta)\gamma}{\norm{h^{(t)}_1 - \gamma \hn_L(h^{(t)}_1)}} \\
& \geq \frac{1+1.2\theta(1-3\theta)\gamma}{1+\calR/2} \\
& \geq 1 + 1.2\theta(1-3\theta)\gamma - \calR \\
& \geq 1 + \theta(1-3\theta)\gamma,
\end{align*}
where the third inequality follows from the fact that $(1+a)/(1+b/2) > 1+a-b$ if $b>0$ and $0 < a-b < 1$, and the last inequality is due to $\calR < 0.2\theta(1-3\theta)\gamma = \frac{\theta(1-3\theta)}{640n^3}$. 

The proof is completed by applying the above eigenvalue argument $t$ times.
\end{IEEEproof}


\begin{IEEEproof}[Proof of Claim \ref{cla:less_than_half}]
We prove the claim by induction. This claim is clearly true for $t=0$, for which the left-hand side is $0$. Suppose the claim is true for $t = 0, 1,\dots, \calT-1$, and we show that it is true for $t = \calT$.

By the induction hypothesis, for $t = 0, 1,\dots, \calT-1$, we have
\begin{align*}
& \Bigl\| P_{h_0^\perp} (h^{(t)}_1 -h^{(t)}_2) \Bigr\| \\
& \leq 2 \Bigl\|\Bigl[\prod_{\tau = 0}^{t-1}\calP^{(\tau)}\Bigr] P_{h_0^\perp} (h^{(0)}_1 -h^{(0)}_2)\Bigr\| \\
& = 2d \prod_{\tau = 0}^{t-1} \norm{\calP^{(\tau)}},
\end{align*}
and
\begin{align*}
& \norm{q^{(t)}} \leq \gamma \sup \norm{\hH_L(h) - \hH_L(h_0)} \int_{h^{(t)}_1}^{h^{(t)}_2}  \norm{dh} + \sin\angle(h_0, h^{(t+1)}_2) \times \gamma | \norm{\hn_L(h^{(t)}_1)}-\norm{\hn_L(h^{(t)}_2)} | \\
& \leq \gamma \times 272n^3 \calR \times 2 \norm{P_{h_0^\perp} (h^{(t)}_1 -h^{(t)}_2)} + \calR \times \gamma \times 64n^3 \times 2 \norm{P_{h_0^\perp} (h^{(t)}_1 -h^{(t)}_2)} \\
& \leq 6 R \norm{P_{h_0^\perp} (h^{(t)}_1 -h^{(t)}_2)} \\
& \leq 12 d R \prod_{\tau = 0}^{t-1} \norm{\calP^{(\tau)}}.
\end{align*}
It follows that
\begin{align*}
& \Bigl\|\sum_{\tau = 0}^{\calT-1} \Bigl[\prod_{\tau' = \tau+1}^{\calT-1} \calP^{(\tau')}\Bigr] q^{(\tau)}\Bigr\| \\
& \leq \sum_{\tau = 0}^{\calT-1} \Bigl[\prod_{\tau' = \tau+1}^{\calT-1} \norm{\calP^{(\tau')}}\Bigr] \norm{q^{(\tau)}} \\
& \leq 12dR\sum_{\tau = 0}^{\calT-1} \Bigl[\prod_{\tau' = \tau+1}^{\calT-1} \norm{\calP^{(\tau')}}\Bigr] \prod_{\tau' = 0}^{\tau-1} \norm{\calP^{(\tau')}}\\
& \leq 12R \calT \Bigl\|\Bigl[\prod_{\tau = 0}^{\calT-1}\calP^{(\tau)}\Bigr] P_{h_0^\perp} (h^{(0)}_1 -h^{(0)}_2)\Bigr\|.
\end{align*}
The claim follows from our choice of parameters: $12\calR\calT < \frac{1}{2}$.

\end{IEEEproof}


\begin{IEEEproof}[Proof of Lemma \ref{lem:perturbation_escapes_saddle}]
By the construction of the perturbation, $\norm{h^{(0)} - h_0}\leq \sqrt{2-2\sqrt{1-\calR^4}} \leq \frac{\calR}{2}$ since $\calR < 1/4$. By Lemma \ref{lem:coupling_sequence}, the region where $\calT$ iterates of manifold gradient descent get stuck is very narrow. In fact, the probability that 
$L(h^{(0)}) - L(h^{(\calT)}) \leq \calL$ is bounded by
\[
\frac{d \cdot \mathrm{volume}(\calS^{n-2})}{\calD \cdot \mathrm{volume}(\calS^{n-2})} \leq \frac{d\cdot (n-1)}{\calD}  \leq \frac{4\xi n^{(4-\xi/10000)}\log n}{\theta(1-3\theta)}.
\]

Next, we bound $|L(h^{(0)})-L(h_0)|$ using the Lipschitz continuity of the Riemannian gradient:
\begin{align*}
& |L(h^{(0)})-L(h_0)| \\
& \leq \bigl|\langle \hn_L(h_0), h^{(0)} - h_0 \rangle\bigr| + \int_{h_0}^{h^{(0)}} (\hn_L(h) - \hn_L(h_0)) dh \\
& \leq c(n,\theta, \rho) \calD + 64n^3 (2\calD)\times (2\calD) \\
& < \calL / 2.
\end{align*}
Therefore, with probability at least $1 - \frac{4\xi n^{(4-\xi/10000)}\log n}{\theta(1-3\theta)}$,
\[
L(h^{(\calT)}) - L(h_0) = \bigl[ L(h^{(\calT)}) - L(h^{(0)})\Bigr] + \Bigl[ L(h^{(0)}) - L(h_0) \Bigr]
< -\calL + \calL/2 = -\calL/2.
\]
\end{IEEEproof}


\begin{IEEEproof}[Proof of Corollary \ref{cor:f_and_x}]
By Lemma \ref{lem:neighborhood_local_minima}, for $\hat{h} \in \calH_1$, we have $\norm{C_fR \hat{h} \pm e_j}   \leq 2 \sqrt{\rho}$ for some $j\in [n]$. 
Then by the Cauchy-Schwarz inequality
\begin{align}
\norm{\calF(f) \odot \calF(R\hat{h}) - \calF(\mp e_j)}_\infty \leq \sqrt{n} \norm{C_fR \hat{h} \pm e_j} \leq 2 \sqrt{\rho n}.
\label{eq:diagonal_error}
\end{align}
Equivalently, the circular convolution operators satisfy
\[
\norm{C_f C_{R\hat{h}} - C_{\mp e_j}} \leq 2 \sqrt{\rho n}.
\]
It follows that
\begin{align*}
& \norm{\hat{x}_i \pm \calS_j(x_i)} = \norm{C_{y_i}R\hat{h} \pm \calS_j(x_i)}\\
& = \norm{C_f C_{R\hat{h}} x_i - C_{\mp e_j} x_i} \leq \norm{C_f C_{R\hat{h}} - C_{\mp e_j}} \cdot \norm{x_i} \\
& \leq 2 \sqrt{\rho n} \cdot \norm{x_i}.
\end{align*}

It follows from \eqref{eq:diagonal_error} that
\[
|\calF(f)_{(k)} \times \calF(R\hat{h})_{(k)} - e^{\frac{\pm \sqrt{-1} (j-1)(k-1)}{n}} | \leq 2\sqrt{\rho n}, 
\]
for all $k\in [n]$. Therefore,
\[
|\calF(R\hat{h})_{(k)}| \geq \frac{1 - 2 \sqrt{\rho n}}{|\calF(f)_{(k)}|}.
\]
Since $\min_{k\in[n]} |\calF(R\hat{h})_{(k)}| = \sigma_n(C_{R\hat{h}})$, and $\max_{k\in[n]} |\calF(f)_{(k)}| = \norm{\calF(f)}_\infty \leq \sqrt{n}\norm{f}$, we have 
\[
\sigma_n(C_{R\hat{h}}) \geq \frac{1-2 \sqrt{\rho n}}{\sqrt{n} \norm{f}}.
\]
Combining the above bound with the following
\begin{align*}
\norm{C_{R\hat{h}}(f\pm \calS_j(\hat{f}))} = \norm{C_f R\hat{h} \pm e_j} \leq 2 \sqrt{\rho},
\end{align*}
we have
\begin{align*}
& \norm{\hat{f} \pm \calS_{-j}(f)} = \norm{f \pm \calS_j(\hat{f})} \\
& \leq \frac{\norm{C_{R\hat{h}}(f\pm \calS_j(\hat{f}))}}{\sigma_n(C_{R\hat{h}})} \leq 2 \sqrt{\rho} \times \frac{\sqrt{n}\norm{f}}{1-2 \sqrt{\rho n}} \\
& = \frac{2 \sqrt{\rho n}}{1 - 2 \sqrt{\rho n}} \cdot \norm{f},
\end{align*}

\end{IEEEproof}



\begin{thebibliography}{10}
\providecommand{\url}[1]{#1}
\csname url@samestyle\endcsname
\providecommand{\newblock}{\relax}
\providecommand{\bibinfo}[2]{#2}
\providecommand{\BIBentrySTDinterwordspacing}{\spaceskip=0pt\relax}
\providecommand{\BIBentryALTinterwordstretchfactor}{4}
\providecommand{\BIBentryALTinterwordspacing}{\spaceskip=\fontdimen2\font plus
\BIBentryALTinterwordstretchfactor\fontdimen3\font minus
  \fontdimen4\font\relax}
\providecommand{\BIBforeignlanguage}[2]{{%
\expandafter\ifx\csname l@#1\endcsname\relax
\typeout{** WARNING: IEEEtran.bst: No hyphenation pattern has been}%
\typeout{** loaded for the language `#1'. Using the pattern for}%
\typeout{** the default language instead.}%
\else
\language=\csname l@#1\endcsname
\fi
#2}}
\providecommand{\BIBdecl}{\relax}
\BIBdecl

\bibitem{li2018global}
Y.~Li and Y.~Bresler, ``Global geometry of multichannel sparse blind
  deconvolution on the sphere,'' in \emph{Advances in Neural Information
  Processing Systems}, 2018, pp. 1140--1151.

\bibitem{kundur1996blind}
D.~Kundur and D.~Hatzinakos, ``Blind image deconvolution,'' \emph{IEEE Signal
  Processing Magazine}, vol.~13, no.~3, pp. 43--64, May 1996.

\bibitem{cho2009fast}
S.~Cho and S.~Lee, ``Fast motion deblurring,'' in \emph{ACM Transactions on
  Graphics (TOG)}, vol.~28, no.~5.\hskip 1em plus 0.5em minus 0.4em\relax ACM,
  2009, p. 145.

\bibitem{levin2011understanding}
A.~Levin, Y.~Weiss, F.~Durand, and W.~T. Freeman, ``Understanding blind
  deconvolution algorithms,'' \emph{IEEE Transactions on Pattern Analysis and
  Machine Intelligence}, vol.~33, no.~12, pp. 2354--2367, Dec 2011.

\bibitem{xu2013unnatural}
L.~Xu, S.~Zheng, and J.~Jia, ``Unnatural l0 sparse representation for natural
  image deblurring,'' in \emph{Computer Vision and Pattern Recognition (CVPR),
  2013 IEEE Conference on}.\hskip 1em plus 0.5em minus 0.4em\relax IEEE, 2013,
  pp. 1107--1114.

\bibitem{ahmed2014blind}
A.~Ahmed, B.~Recht, and J.~Romberg, ``Blind deconvolution using convex
  programming,'' \emph{IEEE Transactions on Information Theory}, vol.~60,
  no.~3, pp. 1711--1732, March 2014.

\bibitem{ling2015self}
S.~Ling and T.~Strohmer, ``Self-calibration and biconvex compressive sensing,''
  \emph{Inverse Problems}, vol.~31, no.~11, p. 115002, 2015.

\bibitem{chi2016guaranteed}
Y.~Chi, ``Guaranteed blind sparse spikes deconvolution via lifting and convex
  optimization,'' \emph{IEEE Journal of Selected Topics in Signal Processing},
  vol.~10, no.~4, pp. 782--794, June 2016.

\bibitem{li2018rapid}
X.~Li, S.~Ling, T.~Strohmer, and K.~Wei, ``Rapid, robust, and reliable blind
  deconvolution via nonconvex optimization,'' \emph{Applied and computational
  harmonic analysis}, 2018.

\bibitem{lee2017blind}
K.~Lee, Y.~Li, M.~Junge, and Y.~Bresler, ``Blind recovery of sparse signals
  from subsampled convolution,'' \emph{IEEE Transactions on Information
  Theory}, vol.~63, no.~2, pp. 802--821, Feb 2017.

\bibitem{huang2018blind}
W.~Huang and P.~Hand, ``Blind deconvolution by a steepest descent algorithm on
  a quotient manifold,'' \emph{{SIAM} Journal on Imaging Sciences}, vol.~11,
  no.~4, pp. 2757--2785, 2018.

\bibitem{zhang2017global}
Y.~Zhang, Y.~Lau, H.-w. Kuo, S.~Cheung, A.~Pasupathy, and J.~Wright, ``On the
  global geometry of sphere-constrained sparse blind deconvolution,'' in
  \emph{Proceedings of the IEEE Conference on Computer Vision and Pattern
  Recognition}, 2017, pp. 4894--4902.

\bibitem{tong1998multichannel}
L.~Tong and S.~Perreau, ``Multichannel blind identification: from subspace to
  maximum likelihood methods,'' \emph{Proceedings of the IEEE}, vol.~86,
  no.~10, pp. 1951--1968, Oct 1998.

\bibitem{she2015image}
H.~She, R.-R. Chen, D.~Liang, Y.~Chang, and L.~Ying, ``Image reconstruction
  from phased-array data based on multichannel blind deconvolution,''
  \emph{Magnetic resonance imaging}, vol.~33, no.~9, pp. 1106--1113, 2015.

\bibitem{zhang2013multi}
H.~Zhang, D.~Wipf, and Y.~Zhang, ``Multi-image blind deblurring using a coupled
  adaptive sparse prior,'' in \emph{Computer Vision and Pattern Recognition
  (CVPR), 2013 IEEE Conference on}.\hskip 1em plus 0.5em minus 0.4em\relax
  IEEE, 2013, pp. 1051--1058.

\bibitem{tong1991new}
L.~Tong, G.~Xu, and T.~Kailath, ``A new approach to blind identification and
  equalization of multipath channels,'' in \emph{[1991] Conference Record of
  the Twenty-Fifth Asilomar Conference on Signals, Systems Computers}, Nov
  1991, pp. 856--860 vol.2.

\bibitem{moulines1995subspace}
E.~Moulines, P.~Duhamel, J.~F. Cardoso, and S.~Mayrargue, ``Subspace methods
  for the blind identification of multichannel fir filters,'' \emph{IEEE
  Transactions on Signal Processing}, vol.~43, no.~2, pp. 516--525, Feb 1995.

\bibitem{xu1995least}
G.~Xu, H.~Liu, L.~Tong, and T.~Kailath, ``A least-squares approach to blind
  channel identification,'' \emph{IEEE Transactions on Signal Processing},
  vol.~43, no.~12, pp. 2982--2993, Dec 1995.

\bibitem{gurelli1995evam}
M.~I. Gurelli and C.~L. Nikias, ``Evam: an eigenvector-based algorithm for
  multichannel blind deconvolution of input colored signals,'' \emph{IEEE
  Transactions on Signal Processing}, vol.~43, no.~1, pp. 134--149, Jan 1995.

\bibitem{harikumar1998fir}
G.~Harikumar and Y.~Bresler, ``Fir perfect signal reconstruction from multiple
  convolutions: minimum deconvolver orders,'' \emph{IEEE Transactions on Signal
  Processing}, vol.~46, no.~1, pp. 215--218, Jan 1998.

\bibitem{lee2018spectral}
K.~Lee, F.~Krahmer, and J.~Romberg, ``Spectral methods for passive imaging:
  Nonasymptotic performance and robustness,'' \emph{{SIAM} Journal on Imaging
  Sciences}, vol.~11, no.~3, pp. 2110--2164, 2018.

\bibitem{lee2018fast}
K.~{Lee}, N.~{Tian}, and J.~{Romberg}, ``Fast and guaranteed blind multichannel
  deconvolution under a bilinear system model,'' \emph{IEEE Transactions on
  Information Theory}, vol.~64, no.~7, pp. 4792--4818, July 2018.

\bibitem{berger2010sparse}
C.~R. Berger, S.~Zhou, J.~C. Preisig, and P.~Willett, ``Sparse channel
  estimation for multicarrier underwater acoustic communication: From subspace
  methods to compressed sensing,'' \emph{IEEE Transactions on Signal
  Processing}, vol.~58, no.~3, pp. 1708--1721, March 2010.

\bibitem{sabra2004blind}
K.~G. Sabra and D.~R. Dowling, ``Blind deconvolution in ocean waveguides using
  artificial time reversal,'' \emph{The Journal of the Acoustical Society of
  America}, vol. 116, no.~1, pp. 262--271, 2004.

\bibitem{tian2017multichannel}
N.~Tian, S.-H. Byun, K.~Sabra, and J.~Romberg, ``Multichannel myopic
  deconvolution in underwater acoustic channels via low-rank recovery,''
  \emph{The Journal of the Acoustical Society of America}, vol. 141, no.~5, pp.
  3337--3348, 2017.

\bibitem{kaaresen1998multichannel}
K.~F. Kaaresen and T.~Taxt, ``Multichannel blind deconvolution of seismic
  signals,'' \emph{Geophysics}, vol.~63, no.~6, pp. 2093--2107, 1998.

\bibitem{gitelman2003modeling}
D.~R. Gitelman, W.~D. Penny, J.~Ashburner, and K.~J. Friston, ``Modeling
  regional and psychophysiologic interactions in fmri: the importance of
  hemodynamic deconvolution,'' \emph{Neuroimage}, vol.~19, no.~1, pp. 200--207,
  2003.

\bibitem{rust2006sub}
M.~J. Rust, M.~Bates, and X.~Zhuang, ``Sub-diffraction-limit imaging by
  stochastic optical reconstruction microscopy (storm),'' \emph{Nature
  methods}, vol.~3, no.~10, p. 793, 2006.

\bibitem{betzig2006imaging}
E.~Betzig, G.~H. Patterson, R.~Sougrat, O.~W. Lindwasser, S.~Olenych, J.~S.
  Bonifacino, M.~W. Davidson, J.~Lippincott-Schwartz, and H.~F. Hess, ``Imaging
  intracellular fluorescent proteins at nanometer resolution,'' \emph{Science},
  vol. 313, no. 5793, pp. 1642--1645, 2006.

\bibitem{mukamel2012statistical}
E.~A. Mukamel, H.~Babcock, and X.~Zhuang, ``Statistical deconvolution for
  superresolution fluorescence microscopy,'' \emph{Biophysical journal}, vol.
  102, no.~10, pp. 2391--2400, 2012.

\bibitem{sarder2006deconvolution}
P.~Sarder and A.~Nehorai, ``Deconvolution methods for 3-d fluorescence
  microscopy images,'' \emph{IEEE Signal Processing Magazine}, vol.~23, no.~3,
  pp. 32--45, May 2006.

\bibitem{wang2016blind}
L.~Wang and Y.~Chi, ``Blind deconvolution from multiple sparse inputs,''
  \emph{IEEE Signal Processing Letters}, vol.~23, no.~10, pp. 1384--1388, Oct
  2016.

\bibitem{li2017blind}
Y.~Li, K.~Lee, and Y.~Bresler, ``Blind gain and phase calibration via sparse
  spectral methods,'' \emph{IEEE Transactions on Information Theory}, 2018.

\bibitem{strohmer2002four}
T.~Strohmer, ``Four short stories about toeplitz matrix calculations,''
  \emph{Linear Algebra and its Applications}, vol. 343, pp. 321--344, 2002.

\bibitem{li2017identifiability}
Y.~Li, K.~Lee, and Y.~Bresler, ``Identifiability in bilinear inverse problems
  with applications to subspace or sparsity-constrained blind gain and phase
  calibration,'' \emph{IEEE Transactions on Information Theory}, vol.~63,
  no.~2, pp. 822--842, Feb 2017.

\bibitem{balzano2007blind}
L.~Balzano and R.~Nowak, ``Blind calibration of sensor networks,'' in
  \emph{Proceedings of the 6th international conference on Information
  processing in sensor networks}.\hskip 1em plus 0.5em minus 0.4em\relax ACM,
  2007, pp. 79--88.

\bibitem{bilen2014convex}
C.~Bilen, G.~Puy, R.~Gribonval, and L.~Daudet, ``Convex optimization approaches
  for blind sensor calibration using sparsity,'' \emph{IEEE Transactions on
  Signal Processing}, vol.~62, no.~18, pp. 4847--4856, Sept 2014.

\bibitem{ling2018self}
S.~Ling and T.~Strohmer, ``Self-calibration and bilinear inverse problems via
  linear least squares,'' \emph{{SIAM} Journal on Imaging Sciences}, vol.~11,
  no.~1, pp. 252--292, jan 2018.

\bibitem{sun2017complete}
J.~Sun, Q.~Qu, and J.~Wright, ``Complete dictionary recovery over the sphere i:
  Overview and the geometric picture,'' \emph{IEEE Transactions on Information
  Theory}, vol.~63, no.~2, pp. 853--884, Feb 2017.

\bibitem{sun2017geometric}
\BIBentryALTinterwordspacing
------, ``A geometric analysis of phase retrieval,'' \emph{Foundations of
  Computational Mathematics}, Aug 2017.
\BIBentrySTDinterwordspacing

\bibitem{mei2016landscape}
S.~Mei, Y.~Bai, and A.~Montanari, ``The landscape of empirical risk for
  non-convex losses,'' \emph{arXiv preprint arXiv:1607.06534}, 2016.

\bibitem{sun2017complete2}
J.~Sun, Q.~Qu, and J.~Wright, ``Complete dictionary recovery over the sphere
  ii: Recovery by riemannian trust-region method,'' \emph{IEEE Transactions on
  Information Theory}, vol.~63, no.~2, pp. 885--914, Feb 2017.

\bibitem{lee2017first}
J.~D. Lee, I.~Panageas, G.~Piliouras, M.~Simchowitz, M.~I. Jordan, and
  B.~Recht, ``First-order methods almost always avoid saddle points,''
  \emph{arXiv preprint arXiv:1710.07406}, 2017.

\bibitem{jin2019stochastic}
C.~Jin, P.~Netrapalli, R.~Ge, S.~M. Kakade, and M.~I. Jordan, ``Stochastic
  gradient descent escapes saddle points efficiently,'' \emph{arXiv preprint
  arXiv:1902.04811}, 2019.

\bibitem{lee2016gradient}
J.~D. Lee, M.~Simchowitz, M.~I. Jordan, and B.~Recht, ``Gradient descent only
  converges to minimizers,'' in \emph{Conference on Learning Theory}, 2016, pp.
  1246--1257.

\bibitem{panageas2016gradient}
I.~Panageas and G.~Piliouras, ``Gradient descent only converges to minimizers:
  Non-isolated critical points and invariant regions,'' \emph{arXiv preprint
  arXiv:1605.00405}, 2016.

\bibitem{ge2015escaping}
R.~Ge, F.~Huang, C.~Jin, and Y.~Yuan, ``Escaping from saddle points -- online
  stochastic gradient for tensor decomposition,'' in \emph{Conference on
  Learning Theory}, 2015, pp. 797--842.

\bibitem{jin2017escape}
C.~Jin, R.~Ge, P.~Netrapalli, S.~M. Kakade, and M.~I. Jordan, ``How to escape
  saddle points efficiently,'' in \emph{International Conference on Machine
  Learning}, 2017, pp. 1724--1732.

\bibitem{allen2017natasha}
Z.~Allen-Zhu, ``Natasha: Faster non-convex stochastic optimization via strongly
  non-convex parameter,'' in \emph{Proceedings of the 34th International
  Conference on Machine Learning-Volume 70}.\hskip 1em plus 0.5em minus
  0.4em\relax JMLR. org, 2017, pp. 89--97.

\bibitem{allen2018natasha}
------, ``Natasha 2: Faster non-convex optimization than sgd,'' in
  \emph{Advances in Neural Information Processing Systems}, 2018, pp.
  2680--2691.

\bibitem{agarwal2017finding}
N.~Agarwal, Z.~Allen-Zhu, B.~Bullins, E.~Hazan, and T.~Ma, ``Finding
  approximate local minima faster than gradient descent,'' in \emph{Proceedings
  of the 49th Annual ACM SIGACT Symposium on Theory of Computing}.\hskip 1em
  plus 0.5em minus 0.4em\relax ACM, 2017, pp. 1195--1199.

\bibitem{goldfarb2017using}
\BIBentryALTinterwordspacing
D.~Goldfarb, C.~Mu, J.~Wright, and C.~Zhou, ``Using negative curvature in
  solving nonlinear programs,'' \emph{Computational Optimization and
  Applications}, vol.~68, no.~3, pp. 479--502, Dec 2017.
\BIBentrySTDinterwordspacing

\bibitem{chen2018gradient}
Y.~Chen, Y.~Chi, J.~Fan, and C.~Ma, ``Gradient descent with random
  initialization: Fast global convergence for nonconvex phase retrieval,''
  \emph{Mathematical Programming}, pp. 1--33, 2018.

\bibitem{shalvi1990new}
O.~Shalvi and E.~Weinstein, ``New criteria for blind deconvolution of
  nonminimum phase systems (channels),'' \emph{IEEE Transactions on Information
  Theory}, vol.~36, no.~2, pp. 312--321, March 1990.

\bibitem{donoho2003optimally}
D.~L. Donoho and M.~Elad, ``Optimally sparse representation in general
  (nonorthogonal) dictionaries via $\ell_1$ minimization,'' \emph{Proceedings
  of the National Academy of Sciences}, vol. 100, no.~5, pp. 2197--2202, feb
  2003.

\bibitem{zhang2017structured}
Y.~Zhang, H.-W. Kuo, and J.~Wright, ``Structured local optima in sparse blind
  deconvolution,'' in \emph{Proceedings of the 10th NIPS Workshop on
  Optimization for Machine Learning (OPTML)}, 2017.

\bibitem{bai2018subgradient}
Y.~Bai, Q.~Jiang, and J.~Sun, ``Subgradient descent learns orthogonal
  dictionaries,'' \emph{arXiv preprint arXiv:1810.10702}, 2018.

\bibitem{beck2009fast}
A.~Beck and M.~Teboulle, ``A fast iterative shrinkage-thresholding algorithm
  for linear inverse problems,'' \emph{SIAM journal on imaging sciences},
  vol.~2, no.~1, pp. 183--202, 2009.

\bibitem{absil2009optimization}
P.-A. Absil, R.~Mahony, and R.~Sepulchre, \emph{Optimization algorithms on
  matrix manifolds}.\hskip 1em plus 0.5em minus 0.4em\relax Princeton
  University Press, 2009.

\bibitem{boumal2018global}
N.~Boumal, P.-A. Absil, and C.~Cartis, ``Global rates of convergence for
  nonconvex optimization on manifolds,'' \emph{{IMA} Journal of Numerical
  Analysis}, vol.~39, no.~1, pp. 1--33, 2018.

\bibitem{candes2015phase}
E.~J. Cand{\`{e}}s, X.~Li, and M.~Soltanolkotabi, ``Phase retrieval via
  wirtinger flow: Theory and algorithms,'' \emph{IEEE Transactions on
  Information Theory}, vol.~61, no.~4, pp. 1985--2007, April 2015.

\bibitem{wylie1993self}
M.~P. Wylie, S.~Roy, and R.~F. Schmitt, ``Self-calibration of linear
  equi-spaced (les) arrays,'' in \emph{1993 IEEE International Conference on
  Acoustics, Speech, and Signal Processing}, vol.~1, April 1993, pp. 281--284
  vol.1.

\bibitem{paulraj1985direction}
A.~Paulraj and T.~Kailath, ``Direction of arrival estimation by eigenstructure
  methods with unknown sensor gain and phase,'' in \emph{ICASSP '85. IEEE
  International Conference on Acoustics, Speech, and Signal Processing},
  vol.~10, Apr 1985, pp. 640--643.

\bibitem{eldar2018sensor}
Y.~C. Eldar, W.~Liao, and S.~Tang, ``Sensor calibration for off-the-grid
  spectral estimation,'' \emph{Applied and Computational Harmonic Analysis},
  2018.

\bibitem{yuan2013truncated}
X.-T. Yuan and T.~Zhang, ``Truncated power method for sparse eigenvalue
  problems,'' \emph{Journal of Machine Learning Research}, vol.~14, no. Apr,
  pp. 899--925, 2013.

\bibitem{tropp2011user}
J.~A. Tropp, ``User-friendly tail bounds for sums of random matrices,''
  \emph{Foundations of Computational Mathematics}, vol.~12, no.~4, pp.
  389--434, aug 2011.

\bibitem{ledoux2013probability}
M.~Ledoux and M.~Talagrand, \emph{Probability in Banach Spaces: isoperimetry
  and processes}.\hskip 1em plus 0.5em minus 0.4em\relax Springer Science \&
  Business Media, 2013.

\bibitem{sun2015complete}
J.~Sun, Q.~Qu, and J.~Wright, ``Complete dictionary recovery over the sphere,''
  \emph{arXiv preprint arXiv:1504.06785}, 2015.

\bibitem{higham2008functions}
N.~J. Higham, \emph{Functions of matrices: theory and computation}.\hskip 1em
  plus 0.5em minus 0.4em\relax {SIAM}, 2008, vol. 104.

\end{thebibliography}

\end{document}